\documentclass[11pt]{article}
\usepackage[utf8]{inputenc}
\usepackage{amsmath, amssymb, amsfonts}
\usepackage{amsthm}
\usepackage{enumitem}
\usepackage{bm}
\usepackage{graphicx}
\usepackage{overpic}
\usepackage{geometry}
\usepackage{setspace}
\makeatletter
\@addtoreset{equation}{section}

\makeatother
\theoremstyle{definition}
\newtheorem{theorem}{Theorem}[section]

\newtheorem{lemma}[theorem]{Lemma}
\newtheorem{proposition}[theorem]{Proposition}
\newtheorem{claim}{Claim}
\newtheorem{remark}[theorem]{Remark}

\begin{document}

\title{Minimum 0-Extension Problems
on Directed Metrics}

\author{Hiroshi HIRAI and Ryuhei MIZUTANI\\
Department of Mathematical Informatics,\\
Graduate School of Information Science and Technology,\\
The University of Tokyo, Tokyo, 113-8656, Japan.\\
\texttt{hirai@mist.i.u-tokyo.ac.jp}\\
\texttt{ryuhei\_mizutani@mist.i.u-tokyo.ac.jp}
}
\date{February 7, 2021}

\maketitle

\begin{abstract}
For a metric $\mu$ on a finite set $T$, the minimum 0-extension problem \textbf{0-Ext}$[\mu]$ is defined as follows: Given $V\supseteq T$ and $\ c:{\scriptsize\begin{pmatrix}V \\
                     2 \\
\end{pmatrix}}\rightarrow \mathbb{Q}_+$, minimize $\sum c(xy)\mu(\gamma(x),\gamma(y))$ subject to $\gamma:V\rightarrow T,\ \gamma(t)=t\ (\forall t\in T)$, where the sum is taken over all unordered pairs in $V$. This problem generalizes several classical combinatorial optimization problems such as the minimum cut problem or the multiterminal cut problem. 
Karzanov and Hirai established a complete classification of metrics $\mu$ for which \textbf{0-Ext}$[\mu]$ is polynomial time solvable or NP-hard. This result can also be viewed as a sharpening of the general dichotomy theorem for finite-valued CSPs (Thapper and \v{Z}ivn\'{y} 2016) specialized to \textbf{0-Ext}$[\mu]$.

In this paper, we consider a directed version
$\overrightarrow{\textbf{0}}$\textbf{-Ext}$[\mu]$ of the minimum 0-extension problem,
where $\mu$ and $c$ are not assumed to be symmetric.
We extend the NP-hardness condition of \textbf{0-Ext}$[\mu]$ 
to $\overrightarrow{\textbf{0}}$\textbf{-Ext}$[\mu]$: If $\mu$ cannot be represented as the
shortest path metric of
an orientable modular graph with an orbit-invariant ``directed'' edge-length,
then $\overrightarrow{\textbf{0}}$\textbf{-Ext}$[\mu]$ is NP-hard.
We also show a partial converse: If $\mu$ is a directed metric of
a modular lattice with an orbit-invariant directed edge-length,
then $\overrightarrow{\textbf{0}}$\textbf{-Ext}$[\mu]$ is tractable.
We further provide a new NP-hardness condition characteristic
of $\overrightarrow{\textbf{0}}$\textbf{-Ext}$[\mu]$, and establish
a dichotomy for the case where $\mu$ is a directed metric of a star.
\end{abstract}

\section{Introduction}
\label{sec:introduction}
A \textit{metric} on a finite set $T$ is a function $\mu:T\times T\rightarrow \mathbb{R}_{+} $ that satisfies $\mu(x,x)=0,\ \mu(x,y)=\mu(y,x),$ and $\mu(x,y)+\mu(y,z)\geq \mu(x,z)$ for every $x,y,z\in T$, and $\mu(x,y)>0$ for every $x\neq y\in T$. 
For a rational-valued metric  $\mu$ on $T$, the \textit{minimum} 0-\textit{extension problem} \textbf{0-Ext}$[\mu]$ on $\mu$ is defined as follows:
\begin{align}
\label{0ext}
\textbf{0-Ext}[\mu]\mathrm{:\ \ \ \ }
\mathrm{Instance:}&\mathrm{\ \ }V\supseteq T,\ c:\begin{pmatrix}V \\ 2 \\ \end{pmatrix}
\rightarrow \mathbb{Q}_{+}\notag\\
\mathrm{Min.}&\displaystyle\sum_{\tiny xy\in
\begin{pmatrix}V \\2 \\
\end{pmatrix}\normalsize}c(xy)\mu(\gamma(x),\gamma(y)) \notag \\ \mathrm{s.t.}&\mathrm{\ \ }\gamma:V\rightarrow T\  \mathrm{with\ }\gamma(t)=t\mathrm{\ for\ all\ }t\in T,\mathrm{\ \ \ \ \ \ \ \ \ \ \ \ \ \ \ \ \ \ \ \ \ \ \ \ \ \ \ \ \ \ \ \ \ }
\end{align}
where ${\scriptsize\begin{pmatrix}V \\
                     2 \\
      \end{pmatrix}}$ denotes the set of all unordered pairs of $V$,
and $xy$ denotes the unordered pair consisting of $x,y\in V$. The minimum 0-extension problem was introduced by Karzanov \cite{karzanov1998}, and also known as the \textit{multifacility location problem} in facility location theory \cite{tansel1983}. Note that the formulation (\ref{0ext}) of \textbf{0-Ext}$[\mu]$ is different from but equivalent to that of \cite{karzanov1998}.

The minimum 0-extension problem generalizes several classical combinatorial optimization problems: If $T=\{s,t\}$, then \textbf{0-Ext}$[\mu]$ is nothing but the minimum $s$-$t$ cut problem in an undirected network. If $T=\{x,y,z\}$ and $\mu(x,y)=\mu(y,z)=\mu(z,x)=1$, then \textbf{0-Ext}$[\mu]$ is the $3$-terminal cut problem. Similarly, \textbf{0-Ext}$[\mu]$ can formulate the $k$-terminal cut problem. Moreover, \textbf{0-Ext}$[\mu]$ appears as a discretized LP-dual problem for a class of maximum multiflow problems \cite{karzanov1998pri,karzanov1998} (also see \cite{hirai2009,hirai2011}). 

The computational complexity of \textbf{0-Ext}$[\mu]$ depends on metric $\mu$. In the above examples, the minimum $s$-$t$ cut problem is in P and the $3$-terminal cut problem is NP-hard.
In~\cite{karzanov1998}, Karzanov addressed the classification problem of the computational complexity of \textbf{0-Ext}$[\mu]$ with respect to $\mu$. After \cite{chepoi1996,karzanov2004one}, the complexity dichotomy of \textbf{0-Ext}$[\mu]$ was fully established by Karzanov~\cite{karzanov2004} and Hirai \cite{hirai2016}, which we explain below.

A metric $\mu$ on $T$ is called \textit{modular} if for every $s_0,s_1,s_2\in T$, there exists an element $m\in T$, called a \textit{median}, such that $\mu(s_i,s_j)=\mu(s_i,m)+\mu(m,s_j)$ holds for every $0\leq 
i<j\leq2$. The \textit{underlying graph} of $\mu$ is defined as the undirected graph $H_\mu=(T,U)$, where 
$U=\{\{x,y\}\mid x,y\in T\ (x\neq y),\ \forall z\in T\setminus \{x,y\},\  \mu(x,y)<\mu(x,z)+\mu(z,y)\}$. We say that an undirected graph is \textit{orientable} if it has an edge-orientation such that for every 4-cycle $(u,v,w,z,u)$, the edge $\{u,v\}$ is oriented from $u$ to $v$ if and only if the edge $\{w,z\}$ is oriented from $z$ to $w$.

The dichotomy theorem of the minimum 0-extension problem is the following:
\begin{theorem}[\cite{karzanov2004}]
\label{thm:undirected-nph}
\textit{Let} $\mu$ \textit{be a rational-valued metric. }\textbf{0-Ext}$[\mu]$ \textit{is strongly NP-hard\footnote[1]{A problem $\mathcal{P}$ is called \textit{strongly NP-hard} if $\mathcal{P}$ is still NP-hard when all numbers of the instance are bounded by some polynomial in the length of the instance.} if}
\begin{enumerate}[label=(\roman*),ref=\roman*]
    \item $\mu$ \textit{is not modular, or}
    \label{dichotomy:not modular}
    \item $H_\mu$ \textit{is not orientable}.
    \label{dichotomy: not orientable}
\end{enumerate}
\end{theorem}
\begin{theorem}[\cite{hirai2016}]
\label{thm:undirected-p}
\textit{Let }$\mu$ \textit{be a rational-valued metric. If }$\mu$ \textit{is modular and} $H_\mu$ \textit{is orientable, then }\textbf{0-Ext}$[\mu]$ \textit{is solvable in polynomial time.}
\end{theorem}

The minimum 0-extension problem
constitutes a fundamental class of \textit{valued CSPs} (\textit{valued constraint satisfaction problem}) \cite{hirai2016} --- a minimization problem of a sum of functions having a constant number of variables. More concretely, \textbf{0-Ext}$[\mu]$ is precisely the \textit{finite-valued CSP} generated by a single binary function $\mu:T\times T\rightarrow \mathbb{Q}_+$ that is a metric. Thapper and \v{Z}ivn\'{y} \cite{thapper2016} established a P-or-NP-hard dichotomy theorem for finite-valued CSPs in terms of a certain fractional polymorphism, and moreover, gave a polynomial time algorithm for the \textit{meta problem} of deciding whether a given \textit{template} (meaning $\mu$ in the case of \textbf{0-Ext}$[\mu]$) defines tractable valued CSPs. Their framework is so general and powerful as to be applicable to \textbf{0-Ext}$[\mu]$, which also provides the complexity dichotomy property of \textbf{0-Ext}$[\mu]$ and a polynomial time checking of a tractable metric $\mu$ for \textbf{0-Ext}$[\mu]$. On the other hand, it does not unravel which combinatorial or geometric properties of metric $\mu$ are connected to the tractability of \textbf{0-Ext}$[\mu]$. Although the proof of Theorem \ref{thm:undirected-p} also utilized a related polymorphism condition obtained by Kolmogorov, Thapper, and \v{Z}ivn\'{y} \cite{kolmogorov2015}, it required a deep study on modular metric spaces to verify this condition. In addition to geometric insights, such an explicit tractability characterization as in Theorems \ref{thm:undirected-nph} and \ref{thm:undirected-p} brings an efficient combinatorial algorithm checking the \textbf{0-Ext}-tractability of metric $\mu$, i.e., the meta problem for \textbf{0-Ext}$[\mu]$.
Indeed, we can verify modularity of $\mu$ by checking whether $m$ is a median of triple $t_1,t_2,t_3$ for every $m,t_1,t_2,t_3\in T$. We can also verify the orientability of $H_\mu$ by orienting each edge in depth first order with respect to an adjacency relation such that edges $\{u,v\}$ and $\{z,w\}$ in each 4-cycle $(u,v,w,z,u)$ are said to be adjacent. It is an $O(|T|^4)$-time algorithm. On the other hand, a polynomial time algorithm obtained by the framework of \cite{thapper2016} requires repeated uses of the ellipsoid method. So it is a natural direction to seek such an efficient combinatorial characterization for a more general binary function $\mu:T\times T\rightarrow \mathbb{Q}_+$ for which the corresponding valued CSP is tractable.

Motivated by these facts, in this paper, we consider a \textit{directed} version of the minimum 0-extension problem, aiming to extend the above results. Here, by ``directed'' we mean that symmetry of $\mu$ and $c$ is not assumed. A \textit{directed metric} on a finite set $T$ is a function $\mu:T\times T\rightarrow \mathbb{R}_+$ that satisfies $\mu(x,x)=0$ and $\mu(x,y)+\mu(y,z)\geq \mu(x,z)$ for every $x,y,z\in T$, and $\mu(x,y)+\mu(y,x)>0$ for every $x\neq y\in T$. For a rational-valued directed metric $\mu$ on $T$, the \textit{directed minimum 0-extension problem} $\overrightarrow{\textbf{0}}$\textbf{-Ext}$[\mu]$ on $\mu$ is defined as follows:
\begin{align*}
\overrightarrow{\textbf{0}}\textbf{-Ext}[\mu]\mathrm{:\ \ \ \ } \mathrm{Instance}:&\mathrm{\ \ }V\supseteq T,\mathrm{\ }c:V\times V
\rightarrow \mathbb{Q}_{+}\\
\mathrm{Min.}&\displaystyle\sum_{\tiny (x,y)\in
V\times V\normalsize}c(x,y)\mu(\gamma(x),\gamma(y))
\\
\mathrm{s.t.}&\mathrm{\ \ }\gamma:V\rightarrow T\mathrm{\ with\ }\gamma(t)=t\mathrm{\ for\ all\ }t\in T.\mathrm{\ \ \ \ \ \ \ \ \ \ \ \ \ \ \ \ \ \ \ \ \ \ \ \ \ \ \ \ \ \ \ \ \ }
\end{align*}
The minimum $s$-$t$ cut problem on a directed network is a typical example of $\overrightarrow{\textbf{0}}$\textbf{-Ext}$[\mu]$ in the case of $T=\{s,t\},\ \mu(s,t)=1,$ and $\mu(t,s)=0$. Also, the directed minimum 0-extension problem contains the undirected version. Hence, the complexity classification of the directed version is an extension of that of the undirected version.

In this paper, we explore sufficient conditions for which $\overrightarrow{\textbf{0}}$\textbf{-Ext}$[\mu]$ is tractable, and for which $\overrightarrow{\textbf{0}}$\textbf{-Ext}$[\mu]$ is NP-hard. Our first contribution is an extension of Theorem \ref{thm:undirected-nph} to the directed version:
\begin{theorem}
\label{thm:directed-nph-extend}
\textit{Let }$\mu$ \textit{be a rational-valued directed metric}. $\overrightarrow{\textbf{0}}$\textbf{-Ext}$[\mu]$ \textit{is strongly NP-hard if one of the following holds:}
\begin{enumerate}
[label=(\roman*),ref=\roman*]
    \item $\mu$ \textit{is not modular.}
    \label{thm:directed-nph-extend-case1}
    \item $H_\mu$ \textit{is not orientable.}
    \label{thm:directed-nph-extend-case2}
    \item $\mu$ \textit{is not directed orbit-invariant.}
    \label{condition:not directed orbit-invariant}
\end{enumerate}
\end{theorem}

The modularity and the underlying graph $H_\mu$ of a directed metric $\mu$ are natural extensions of those of a metric. In \textbf{0-Ext}$[\mu]$, the condition (\ref{dichotomy:not modular}) in Theorem \ref{thm:undirected-nph} contains the condition (\ref{condition:not directed orbit-invariant}) in Theorem \ref{thm:directed-nph-extend}. See Section \ref{sec:directed metric spaces} for the precise definitions of the terminologies. 

We next consider the converse of Theorem \ref{thm:directed-nph-extend}. It is known \cite{bandelt1985} that a canonical example of a modular metric is the graph metric of the covering graph of a modular lattice with respect to an orbit-invariant edge-length. Moreover, a tractable metric $\mu$ in Theorem \ref{thm:undirected-p} is obtained by \textit{gluing} such metrics of modular lattices \cite{hirai2016}. It turns out in Section \ref{sec:directed metric spaces} that a directed metric excluded by (\ref{thm:directed-nph-extend-case1}), (\ref{thm:directed-nph-extend-case2}), and (\ref{condition:not directed orbit-invariant}) in Theorem \ref{thm:directed-nph-extend} also admits an amalgamated structure of modular lattices.
Our second contribution is the tractability for the building block of such a directed metric.
\begin{theorem}
\label{thm:directed-p}
\textit{Let} $\mu$ \textit{be a rational-valued directed metric. Suppose that} $H_\mu$ \textit{is the covering graph of a modular lattice and} $\mu$ \textit{is directed orbit-invariant. Then }$\overrightarrow{\textbf{0}}$\textbf{-Ext}$[\mu]$ \textit{is solvable in polynomial time.}
\end{theorem}

See Sections \ref{sec:preliminaries} and \ref{sec:directed metric spaces} for the undefined terminologies.

The converse of Theorem \ref{thm:directed-nph-extend} is not true: Even if $H_\mu$ is a tree (that is excluded by (\ref{thm:directed-nph-extend-case1}), (\ref{thm:directed-nph-extend-case2}), and (\ref{condition:not directed orbit-invariant}) in Theorem \ref{thm:directed-nph-extend}), $\overrightarrow{\textbf{0}}$\textbf{-Ext}$[\mu]$ can be NP-hard. On the other hand, \textbf{0-Ext}$[\mu]$ for which $H_\mu$ is a tree is always tractable (see \cite{tansel1983}). This is a notable difference between \textbf{0-Ext}$[\mu]$ and $\overrightarrow{\textbf{0}}$\textbf{-Ext}$[\mu]$. Our third contribution is a new hardness condition capturing this difference. For $x,y\in T$, let $I_\mu(x,y):=\{z\in T\mid \mu(x,y)=\mu(x,z)+\mu(z,y)\}$, which is called the \textit{interval} from $x$ to $y$. We denote $I:=I_\mu$ if $\mu$ is clear in the context. For $x,y\in T$, the \textit{ratio} $R_{\mu}(x,y)$ from $x$ to $y$ is defined by $R_{\mu}(x,y):=\mu(x,y)/\mu(y,x)$ (if $\mu(y,x)=0$, then $R_{\mu}(x,y):=\infty$). A pair $(x,y)\in
{\scriptsize\begin{pmatrix}T \\
                     2 \\
      \end{pmatrix}}$
is called a \textit{biased pair} if $R_{\mu}(x,z)>R_{\mu}(z,y)$ holds for every $z\in I(x,y)\cap I(y,x)\setminus \{x,y\}$, or $R_{\mu}(x,z)<R_{\mu}(z,y)$ holds for every $z\in I(x,y)\cap I(y,x)\setminus \{x,y\}$. A triple $(s_0,s_1,s_2)$ is called a \textit{non-collinear triple} if $s_i\notin I(s_{i-1},s_{i+1})\cap I(s_{i+1},s_{i-1})$ holds for every $i\in \{0,1,2\}$ (the indices of $s_i$ are taken modulo 3). A non-collinear triple $(s_0,s_1,s_2)$ is also called a \textit{biased non-collinear triple} if $(s_i,s_j)$ is a biased pair for every $i\neq j$. We now state an additional NP-hardness condition of $\overrightarrow{\textbf{0}}$\textbf{-Ext}$[\mu]$:
\begin{theorem}
\label{thm:directed-nph-new}
\textit{Let} $\mu$ \textit{be a rational-valued directed metric on }$T$. \textit{If there exists a biased non-collinear triple for }$\mu$\textit{, then }$\overrightarrow{\textbf{0}}$\textbf{-Ext}$[\mu]$ \textit{is strongly NP-hard.}
\end{theorem}

\begin{figure}[tbp]
\begin{center}
\begin{overpic}[width=14cm]{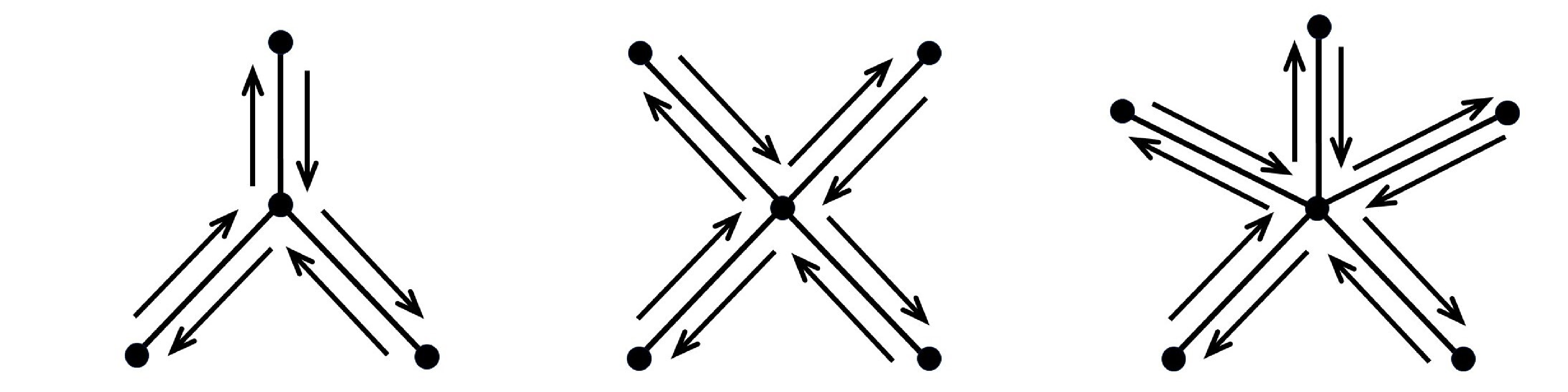}
\put(10,8){1}
\put(14.5,4){2}
\put(20,3.7){1}
\put(23.5,9){2}
\put(14,15.5){2}
\put(20.5,15.5){1}
\put(15.6,11.3){$r$}
\put(5.5,1.5){$x$}
\put(17,24.5){$y$}
\put(28.8,1.5){$z$}
\put(42,8){1}
\put(47,4){2}
\put(56,8.5){2}
\put(52,3.5){4}
\put(43,13){1}
\put(46.7,18.5){3}
\put(56.5,13.7){2}
\put(51.7,18){6}
\put(47.6,11){$r$}
\put(37.5,1.5){$x$}
\put(38,22){$y$}
\put(60.8,21.5){$z$}
\put(60.8,1.5){$w$}
\put(76,8){1}
\put(80.7,4){2}
\put(90.2,8.5){2}
\put(86.1,3.5){4}
\put(75.5,11.5){1}
\put(77,17){1}
\put(80.5,17){3}
\put(86,17){3}
\put(92,11.8){2}
\put(89,16.8){2}
\end{overpic}
\caption{(a)\ NP-hard\ case\ \ \ \ \ \ \ \ (b)\ tractable\ case\ \ \ \ \ \ \ \ \ \ \ \ (c)\ tractable\ case\ \ \ \ \ }
\end{center}
\end{figure}
Figure 1 (a) is an example of a directed metric satisfying the condition in Theorem \ref{thm:directed-nph-new}, as described below. Suppose that $\mu$ is the directed metric such that the underlying graph $H_\mu$ of $\mu$ is the undirected graph described in Figure~1~(a). Then, we have $\mu(v,r)=1,\ \mu(r,v)=2$ for every $v\in \{x,y,z\}$. Since $R_\mu(x,r)=1/2$ and $R_\mu(x,y)=3/3=1$, $(x,y)$ is a biased pair. Similarly, $(y,z)$ and $(z,x)$ are biased pairs. Hence, $(x,y,z)$ is a biased non-collinear triple. 

Our fourth contribution says that the non-existence of a biased non-collinear triple implies tractability, provided the underlying graph is a star.
\begin{theorem}
\label{thm:directed-p-star}
\textit{Let} $\mu$ \textit{be a rational-valued directed metric on $T$. If }$H_\mu$ \textit{is a star and there exists no biased non-collinear triple for $\mu$, then }$\overrightarrow{\textbf{0}}$\textbf{-Ext}$[\mu]$ \textit{is solvable in polynomial time.}
\end{theorem}

Figure 1 (b) is an example of a directed metric satisfying the condition in Theorem \ref{thm:directed-p-star}, as $(x,w)$ and $(y,z)$ are not biased pairs in the directed metric of Figure 1 (b) (the directed metric of Figure 1 (b) is defined in the same way as that of Figure 1 (a)). The directed metric of Figure 1 (c) also satisfies the condition in Theorem \ref{thm:directed-p-star}.

The organization of this paper is as follows. Section \ref{sec:preliminaries} provides preliminary arguments which are necessary for the proofs. Section \ref{sec:directed metric spaces} introduces some notions and shows several properties in directed metric spaces. Section \ref{sec:proof of tractability} proves the tractability results of $\overrightarrow{\textbf{0}}$\textbf{-Ext}$[\mu]$.
To show Theorem~\ref{thm:directed-p}~and~\ref{thm:directed-p-star}, we utilize the tractability condition of valued CSPs by Kolmogorov, Thapper, and \v{Z}ivn\'{y}~\cite{kolmogorov2015}, as was utilized in \cite{hirai2016} to prove Theorem \ref{thm:undirected-p}. Section \ref{sec:proof of hardness} shows the hardness results of $\overrightarrow{\textbf{0}}$\textbf{-Ext}$[\mu]$. We prove Theorem \ref{thm:directed-nph-extend} and \ref{thm:directed-nph-new} by use of the \textit{MC condition} in \cite{cohen2006}, which is similar to the hardness proof of the multiterminal cut problem \cite{dahlhaus1994}, and the proof of Theorem \ref{thm:undirected-nph} in \cite{karzanov1998,karzanov2004}. 

A preliminary version \cite{hirai2020} of this paper has appeared in Proceedings of the 45th International Symposium on Mathematical Foundations of Computer Science (MFCS 2020).

\paragraph*{Notation.} Let $\mathbb{R},\ \mathbb{Q},$ $\mathbb{Z}$, and $\mathbb{N}$ denote the sets of reals, rationals, integers, and positive integers, respectively. We also denote the sets of nonnegative reals and rationals by $\mathbb{R}_+$ and $\mathbb{Q}_+$.

\section{Preliminaries}
\label{sec:preliminaries}
\subsection{Valued CSP}
\label{subsec:valued csp}
Let $D$ be a finite set. For a positive integer $k$, a function $f:D^k \rightarrow \mathbb{Q}$ is called a $k$\textit{-ary cost function} on $D$. Let $ar(f):=k$ denote the arity of $f$. 
The \textit{valued constraint satisfaction problem} (\textit{VCSP}) on $D$ is defined as follows \cite{zivny2012}:
\begin{align*}
\textbf{VCSP:}\mathrm{\ \ \ \ } \mathrm{Instance}:&\mathrm{\ \ }\mathrm{a\ set\ }V=\{x_1,x_2,\ldots,x_n\}\mathrm{\ of\ variables},\\
&\mathrm{\ \ rational\mathchar`-valued\ cost\ functions\ }f_1,f_2,\ldots,f_q\ \mathrm{on\ }D,\mathrm{\ }\\
&\mathrm{\ \ weights\ }w_1,w_2,\ldots,w_q\in \mathbb{Q}_+\ \mathrm{of\ cost\ functions},\\
&\mathrm{\ \ assignments\ }\sigma_i:\{1,2,\ldots,ar(f_i)\}\rightarrow \{1,2,\ldots,n\}\ (i=1,2,\ldots,q)\\
\mathrm{Min.}&\ \displaystyle\sum_{i=1}^{q}w_i\cdot f_i\left(x_{\sigma_i(1)},x_{\sigma_i(2)},\ldots,x_{\sigma_i(ar(f_i))}\right)\\ \mathrm{s.t.}&\mathrm{\ \ }(x_1,x_2,\ldots,x_n)\in D^n.
\end{align*}
In an instance of VCSP, cost functions are extensionally represented; we are given all possible values of all cost functions. Hence, the size of an instance is exponential in the arity of each cost function.

A set of cost functions on $D$ is called a \textit{constraint language} on $D$. Unless specifically said otherwise, we assume that all constraint languages are finite.
Let $\Gamma$ be a constraint language on $D$. The instance of VCSP is called a $\Gamma$\textit{-instance} if all cost functions in the instance belong to $\Gamma$. 
Let VCSP$[\Gamma]$ denote the class of the optimization problems whose instances are restricted to $\Gamma$-instances.

Let $\mu$ be a directed metric on $T$. The directed minimum 0-extension problem $\overrightarrow{\textbf{0}}$\textbf{-Ext}$[\mu]$ is viewed as a subclass of VCSP. Indeed, let
\begin{align}
    D&:=T,\notag\\
    f(x,y)&:=\mu(x,y),\notag\\
    g_t(x)&:=\mu(x,t),\notag\\ 
    h_t(x)&:=\mu(t,x), \notag
\end{align}
and let
\begin{align}
\label{0-ext-->vcsp}
    \Gamma:=\{f\}\cup \{g_t\mid t\in T\}\cup \{h_t\mid t\in T\}.
\end{align}
Then we can conclude that $\overrightarrow{\textbf{0}}$\textbf{-Ext}$[\mu]$ is exactly VCSP$[\Gamma]$ on $D$.

Kolmogorov, Thapper, and \v{Z}ivn\'{y} \cite{kolmogorov2015} discovered a criterion for a constraint language $\Gamma$ such that VCSP$[\Gamma]$ is tractable. To describe this, we need some definitions. A function $\varphi:D^m\rightarrow D$ is called an $m$\textit{-ary operation} on $D$. We denote by $\mathcal{O}_D^{(m)}$ the set of $m$-ary operations on $D$. 
A vector $\omega:\mathcal{O}^{(m)}_D\rightarrow \mathbb{R}_+$ is called a \textit{fractional operation} of arity $m$ on $D$ if 
\begin{align}
\label{frac}
    \sum_{\varphi \in \mathcal{O}^{(m)}_D}\omega(\varphi)=1
\end{align}
is satisfied. We denote the support of $\omega$ by
\begin{align}
\label{supp}
    \mathrm{supp}(\omega):=\{\varphi \in \mathcal{O}^{(m)}_D\mid \omega(\varphi)>0\}.
\end{align}
Let $\Gamma$ be a constraint language on $D$. A fractional operation $\omega:\mathcal{O}^{(m)}_D\rightarrow \mathbb{R}_+$ is called a \textit{fractional polymorphism} of $\Gamma$ if for every cost function $f\in \Gamma\ (f:D^n\rightarrow \mathbb{Q})$ and vectors $x^1,\ldots,x^m\in D^n$, 
\begin{align}
\label{polymorphism}
    \sum_{\varphi \in \mathcal{O}^{(m)}_D}\omega(\varphi)f(\varphi(x^1,\ldots,x^m))\leq \frac{1}{m}\sum_{i=1}^mf(x^i),
\end{align}
where $\varphi(x^1,\ldots,x^m)$ indicates
\begin{align}
    \varphi(x^1,\ldots,x^m):=\left(
    \begin{array}{c}
      \varphi(x_1^1,\ldots,x_1^m) \\
      \vdots \\
      \varphi(x_n^1,\ldots,x_n^m)
    \end{array}
  \right).\notag
\end{align}

We now state a powerful theorem on the relation between fractional polymorphisms and tractability of VCSP$[\Gamma]$ by Kolmogorov, Thapper, and \v{Z}ivn\'{y} \cite{kolmogorov2015}:
\begin{theorem}[\cite{kolmogorov2015}]
\label{vcsp}
\textit{Let $\Gamma$ be a constraint language. If}\textit{ $\Gamma$ admits a fractional polymorphism with a semilattice operation in its support},\textit{ then} VCSP$[\Gamma]$ \textit{can be solved in polynomial time.}
\end{theorem}
Here a \textit{semilattice operation} $\varphi$ on $D$ is a binary operation which satisfies $\varphi(x,x)=x,\ \varphi(x,y)=\varphi(y,x)$, and $\varphi(x,\varphi(y,z))=\varphi(\varphi(x,y),z)$ for every $x,y,z\in D$.

\begin{remark}
Thapper and \v{Z}ivn\'{y} \cite{thapper2012} showed that the tractability of finite-valued CSP VCSP$[\Gamma]$ is completely characterized by the existence of a certain fractional polymorphism for $\Gamma$.
The existence of this fractional polymorphism reduces to the feasibility of a linear program over the space of binary fractional polymorphisms.
Via a version of Farkas' lemma, the infeasibility of the LP leads to instances that can solve maximum cut problems. They showed that by using the ellipsoid method (i.e., the machinery of separation-optimization equivalence), the feasibility of this LP can be checked in polynomial time provided the constraint language is a `core'.
Here a constraint language $\Gamma$ is called a `core' if every unary fractional polymorphism $\omega$ of $\Gamma$ consists of injective operations in its support.

Our constraint language $\Gamma$ defined in (\ref{0-ext-->vcsp}) is a `core'. Indeed, $\Gamma$ contains $g_t$ and $h_t$, and hence every operation $\varphi$ in a unary fractional polymorphism must be an identity map (since $g_t(t)+h_t(t)\geq g_t(\varphi(t))+h_t(\varphi(t))$ implies $\varphi(t)=t$).
\end{remark}

Cohen et al. \cite{cohen2006} discovered a sufficient condition for a constraint language $\Gamma$ such that VCSP$[\Gamma]$ is NP-hard. We describe this condition with some definitions.
For a constraint language $\Gamma$, let $\langle \Gamma \rangle$ denote the set of all functions $f(x_1,\ldots,x_m)$ such that for some instance $I$ of VCSP$[\Gamma]$ with objective function $f_I(x_1,\ldots,x_m,x_{m+1},\ldots,x_n)$, we have 
\begin{align}
    f(x_1,\ldots,x_m)=\min_{x_{m+1},\ldots,x_n}f_I(x_1,\ldots,x_m,x_{m+1},\ldots,x_n).
\end{align}
For a constraint language $\Gamma$ on $D$, we define the following condition (MC):
\vspace{0.1in} \\*
\textbf{(MC)} There exist distinct $a,b\in D$ such that $\langle \Gamma \rangle$ contains a binary cost function $f$ satisfying $\mathrm{argmin}\ f=\{(a,b),(b,a)\}$.
\vspace{0.1in} \\*
Cohen et al. \cite{cohen2006} revealed the relation between the condition (MC) and NP-hardness of VCSP$[\Gamma]$ (also see \cite{thapper2016}).
\begin{proposition}[{\cite[Proposition 5.1]{cohen2006}}; see {\cite[Lemma 2]{thapper2016}}]
\label{prop:mcNP-hard}
If a constraint language $\Gamma$ on $D$ satisfies the condition (MC), then VCSP$[\Gamma]$ is strongly NP-hard.
\end{proposition}
Let $\Gamma$ be a constraint language on $D$ satisfying (MC). In \cite{cohen2006}, only the NP-hardness of VCSP$[\Gamma]$ was shown, whereas the ``strongly'' NP-hardness was not explicitly shown. However, the strongly NP-hardness easily follows from the proofs of Proposition 5.1 and Theorem 3.4 in \cite{cohen2006}.  
\begin{remark}
Let $\Gamma$ be a constraint language satisfying (MC). The proof of NP-hardness of VCSP$[\Gamma]$ in \cite{cohen2006} was based on the reduction from the maximum cut problem (MAX CUT). This reduction is an extension of the hardness proofs of the multiterminal cut problem \cite{dahlhaus1994} and \textbf{0-Ext}$[\mu]$ \cite{karzanov1998, karzanov2004} (Theorem \ref{thm:undirected-nph}). We prove NP-hardness of $\overrightarrow{\textbf{0}}$\textbf{-Ext}$[\mu]$ with the aid of NP-hardness of VCSP$[\Gamma]$ in Section \ref{sec:proof of hardness}.
Thapper and \v{Z}ivn\'{y} \cite{thapper2016} also showed the complete characterization of  constraint languages for which finite-valued CSPs are NP-hard, see \cite{thapper2016} for details.
\end{remark}

\subsection{Modular graphs}
Let $G=(V,E)$ be a connected undirected graph. The \textit{graph metric} $d_G:V\times V\rightarrow \mathbb{Z}$ is defined as follows:
\begin{align}
\label{graph metric}
    d_G(x,y):=\mathrm{the\ number\ of\ edges\ in\ a\ shortest\ path\ from}\ x\ \mathrm{to}\ y\ \mathrm{in}\ G\ (x,y\in V).
\end{align}
We denote $d_G$ simply by $d$ if $G$ is clear in the context. We say that $G$ is \textit{modular} if its graph metric $d_G$ is modular (remember that a metric $\mu$ on $T$ is called modular if for every $s_0,s_1,s_2\in T$, there exists a median $m\in T$, which satisfies $\mu(s_i,s_j)=\mu(s_i,m)+\mu(m,s_j)$ for every $0\leq i<j\leq 2$). We show examples of a modular graph and a nonmodular graph in Figure 2.
\begin{figure}[tbp]
\begin{center}
\begin{overpic}[width=170mm]{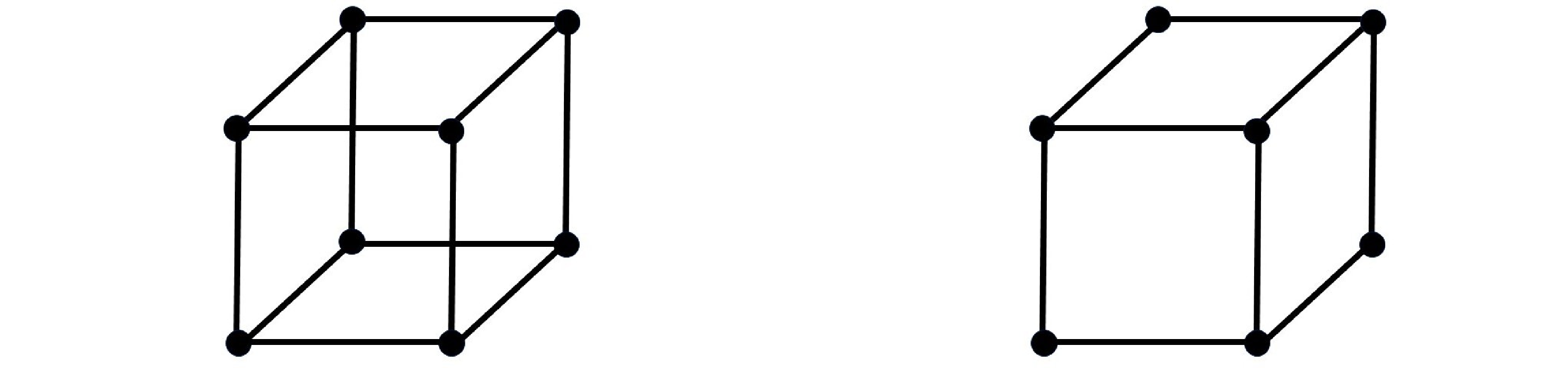}
\put(12,15){$p$}
\put(11.5,1.2){$p_1$}
\put(18.3,22){$p_2$}
\put(37.8,7.3){$q$}
\put(63.6,15){$p$}
\put(63.1,1.2){$p_1$}
\put(69.9,22){$p_2$}
\put(89.4,7.3){$q$}
\put(19.2,7.6){$p'$}
\end{overpic}
\caption{(a)\ a modular graph\ \ \ \ \ \ \ \ \ \ \ \ \ \ \ \ \ \ \ \ \ \ \ \ \ \ \ \ \ \ \ \ (b)\ a nonmodular graph}
\end{center}
\end{figure}
\begin{lemma}[\cite{bandelt1993}]
\label{modularcondition}
\textit{A connected undirected graph $G=(V,E)$ is modular if and only if  the following two conditions hold:}
\begin{enumerate}[label=(\roman*),ref=\roman*]
    \item $G$ \textit{is a bipartite graph}.
    \item \textit{For every pair of vertices} $p,q\in V$ \textit{and neighbors} $p_1,p_2$ \textit{of} $p$ \textit{with} $d(p,q)=1+d(p_1,q)=1+d(p_2,q)$, \textit{there exists a common neighbor} $p'$ \textit{of} $p_1,p_2$ \textit{with} $d(p,q)=2+d(p',q)$.
\end{enumerate}
\end{lemma}
By Lemma \ref{modularcondition}, we can easily check (non)modularity of the graphs in Figure 2. The graph in Figure 2 (b) satisfies $d(p,q)=1+d(p_1,q)=1+d(p_2,q)$. However, there exists no common neighbor $p'$ of $p_1,p_2$ with $d(p,q)=2+d(p',q)$, which implies that this graph is not modular. On the other hand, the graph in Figure 2 (a) has common neighbor $p'$ of $p_1,p_2$ with $d(p,q)=2+d(p',q)$. By similar arguments, we can easily verify that the graph in Figure 2 (a) satisfies the condition (ii) of Lemma \ref{modularcondition}. It is also easy to verify that this graph satisfies the condition (i) of Lemma \ref{modularcondition}. Hence, the graph in Figure 2 (a) is modular.

Let $(T,\mu)$ be a metric space. For $x,y\in T$, we denote the interval of $x,y$ by $I_\mu(x,y):=\{z\in T\mid \mu(x,y)=\mu(x,z)+\mu(z,y)\}$. We denote $I:=I_\mu$ if $\mu$ is clear in the context. A subset $X\subseteq T$ is called a \textit{convex set} if $I(p,q)\subseteq X$ for every $p,q\in X$. A subset $X\subseteq T$ is called a \textit{gated set} if for every $p\in T$, there exists $p'\in X$, called the \textit{gate} of $p$ at $X$, such that $\mu(p,q)=\mu(p,p')+\mu(p',q)$ for every $q\in X$. The gate of $p$ at $X$ is unique. Chepoi \cite{chepoi1989} showed the following relation between convex sets and gated sets:
\begin{lemma}[\cite{chepoi1989}]
\label{gated}
\textit{Let }$G=(V,E)$ \textit{be a modular graph. For the metric space} $(V,d)$ \textit{and a subset} $X\subseteq V$\textit{, the following conditions are equivalent}:
\begin{enumerate}[label=(\roman*),ref=\roman*]
    \item $X$ \textit{is convex.}
    \item $X$ \textit{is gated.}
\end{enumerate}
\end{lemma}

\subsection{Modular lattices}
Let $\mathcal{L}$ be a partially ordered finite set with a partial order $\preceq$. By $a\prec b$ we mean $a\preceq b$ and $a\neq b$. For $a,b\in \mathcal{L}$, we denote by $a\vee b$ the minimum element of the set $\{c\in \mathcal{L}\mid c\succeq a$ and $c\succeq b\}$, and denote by $a\wedge b$ the maximum element of the set $\{c\in \mathcal{L}\mid c\preceq a$ and $c\preceq b\}$. If for every $a,b\in \mathcal{L}$ there exist $a\vee b$ and $a\wedge b$, then $\mathcal{L}$ is called a \textit{lattice}. A lattice $\mathcal{L}$ is called \textit{modular} if for every $a,b,c\in \mathcal{L}$ with $a\preceq c$ it holds that $a\vee (b\wedge c)=(a\vee b)\wedge c$. For $a\preceq b\in \mathcal{L}$, we let $[a,b]$ denote the interval $\{c\in \mathcal{L}\mid a\preceq c \preceq b\}$. For $a\prec b\in \mathcal{L}$, a sequence $(a=u_0,u_1,\ldots,u_n=b)$ is called a \textit{chain} from $a$ to $b$ if $u_{i-1}\prec u_i$ holds for all $i\in \{1,2,\ldots,n\}$. Here the length of a chain $(u_0,u_1,\ldots,u_n)$ is $n$. We denote by $r[a,b]$ the length of the longest chain from $a$ to $b$. For a lattice $\mathcal{L}$, let $\textbf{0}$ denote the minimum element of $\mathcal{L}$, and let $\textbf{1}$ denote the maximum element of $\mathcal{L}$. The rank $r(a)$ of an element $a$ is defined by $r(a):=r[\textbf{0},a]$.
\begin{lemma}[see {\cite[Chapter I\hspace{-.1em}I]{birkhoff1967}}]
\label{jordandedekind}
\textit{Let $\mathcal{L}$ be a modular lattice. For $a\preceq b\in \mathcal{L}$, the following condition }(\textit{called} Jordan--Dedekind chain condition) holds:
\begin{align}
    \mathrm{\textit{All\ maximal\ chains\ from\ }}a \ \mathrm{\textit{to}}\ b\ \mathrm{\textit{have\ the\ same\ length.}}
\end{align}
\end{lemma}
By Lemma \ref{jordandedekind}, we can see that for a modular lattice $\mathcal{L}$ and $a\in \mathcal{L}$, $r(a)$ is equal to the length of a maximal chain from $\textbf{0}$ to $a$. A modular lattice is also characterized by rank as follows:
\begin{lemma}[see {\cite[Chapter I\hspace{-.1em}I]{birkhoff1967}}]
\label{lem:modular}
\textit{A lattice }$\mathcal{L}$ \textit{is modular if and only if for every} $a,b\in \mathcal{L}$, $r(a)+r(b)=r(a\wedge b)+r(a\vee b)$ \textit{holds.}
\end{lemma}
For a poset $\mathcal{L}$ and $a,b\in \mathcal{L}$, we say that $b$ \textit{covers} $a$ if $a\prec b$ holds and there is no $c\in \mathcal{L}$ with $a\prec c\prec b$. The \textit{covering graph} of $\mathcal{L}$ is the undirected graph obtained by linking all pairs $a,b$ of $\mathcal{L}$ such that $a$ covers $b$, or $b$ covers $a$. Here we have the following relation between modular lattices and modular graphs:
\begin{lemma}[\cite{vel1993}]
\label{modular}
\textit{A lattice }$\mathcal{L}$ \textit{is modular if and only if the covering graph of} $\mathcal{L}$ \textit{is modular.}
\end{lemma}

Let $\mathcal{L}$ be a lattice. A function $f:\mathcal{L}\rightarrow \mathbb{R}$ is called \textit{submodular} if $f(p)+f(q)\geq f(p\vee q)+f(p\wedge q)$ holds for every $p,q\in \mathcal{L}$. If $a,b\in \mathcal{L}$ are covered by $a\vee b$, then the pair $(a,b)$ is called a 2-\textit{covered pair}. We have the following characterization of submodular functions on modular lattices:
\begin{lemma}
\label{submodular}
\textit{Let} $\mathcal{L}$ \textit{be a modular lattice. A function} $f:\mathcal{L}\rightarrow \mathbb{R}$ \textit{is submodular if and only if} $f(a)+f(b)\geq f(a\vee b)+f(a\wedge b)$ \textit{holds for every} 2-\textit{covered pair} $(a,b)$.
\end{lemma}
\begin{proof}
The only if part is obvious. We prove the if part. Take $p,q\in \mathcal{L}$ and maximal chains $(p\wedge q=p_0,p_1,\ldots,p_k=p)$ and $(p\wedge q=q_0,q_1,\ldots,q_l=q)$. First we show that $p_{i+1}\vee q_j$ covers $p_i\vee q_j$. Note that $p_{i+1}$ covers $p_i$. Since $\mathcal{L}$ is modular, we have $p_{i+1}\wedge (p_i\vee q_{j})=p_i\vee (q_j\wedge p_{i+1})=p_i$ and$\ p_{i+1}\vee (p_i\vee q_{j})=p_{i+1}\vee q_j$. Hence, we can conclude that $p_{i+1}\vee q_j$ covers $p_i\vee q_j$ by Lemma 2.4. Similarly, $p_i\vee q_{j+1}$ covers $p_i\vee q_j$. Also, by modularity we have $(p_{i+1}\vee q_j)\vee (p_i\vee q_{j+1})=p_{i+1}\vee q_{j+1}$ and $(p_{i+1}\vee q_j)\wedge (p_i\vee q_{j+1})=p_i\vee (q_{j+1}\wedge (p_{i+1}\vee q_j))=p_i\vee q_j\vee (p_{i+1}\wedge q_{j+1})=p_{i}\vee q_{j}$. Then we conclude that $(p_{i+1}\vee q_j,p_i\vee q_{j+1})$ is a 2-covered pair. Let $a_{i,j}:=p_i\vee q_j$. Then we have $f(p)+f(q)-f(p\vee q)-f(p\wedge q)=\sum_{i,j}(f(a_{i+1,j})+f(a_{i,j+1})-f(a_{i+1,j+1})-f(a_{i,j}))\geq 0$. Thus, we conclude that $f$ is submodular.
\end{proof}

\section{Directed metric spaces}
\label{sec:directed metric spaces}
\subsection{Modular directed metrics}
We first extend the notions of modularity, medians, and underlying graphs to directed metric spaces. Let $\mu$ be a directed metric on $T$. We say that $\mu$ is \textit{modular} if and only if for every $s_0,s_1,s_2\in T$, there exists an element $m\in T$, called a \textit{median}, such that $\mu(s_i,s_j)=\mu(s_i,m)+\mu(m,s_j)$ for every $0\leq i,j\leq 2\ (i\neq j)$. See Figure 3 (a) for an example of modular directed metrics. We define the \textit{underlying graph} of $\mu$ as the undirected graph $H_\mu=(T,U)$, where
\begin{align}
    U:=\{\{x,y\} \mid x,y\in T\ (x\neq y),\ \forall z\in T\setminus \{x,y\},\ \mu(x,y)<\mu(x,z)+\mu(z,y)\notag\\ \mathrm{or}\ \forall z\in T\setminus \{x,y\},\ \mu(y,x)<\mu(y,z)+\mu(z,x)\}.
\end{align}
\begin{figure}[tbp]
\begin{center}
\begin{overpic}[width=14cm]{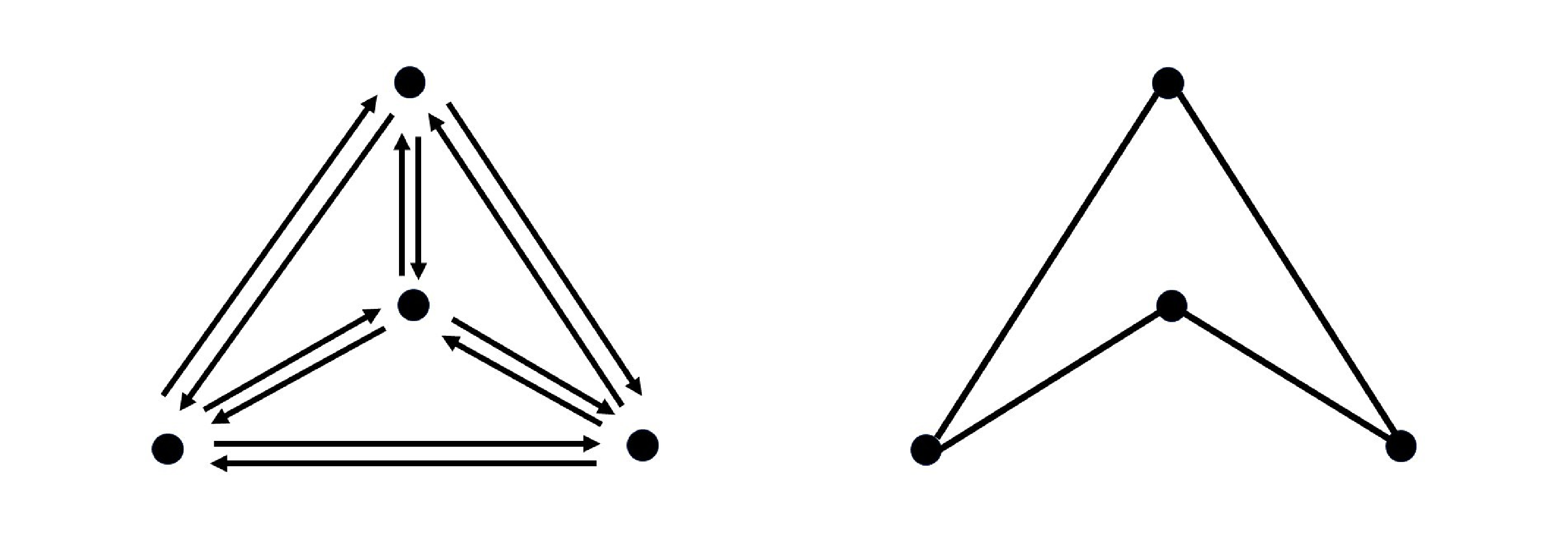}
\put(7.2,5){$p$}
\put(25.5,11.5){$q$}
\put(43,5){$r$}
\put(25.5,31){$s$}
\put(55.5,5){$p$}
\put(74,11.5){$q$}
\put(91,5){$r$}
\put(74,31){$s$}
\put(25.2,1.7){2}
\put(25.2,6.8){3}
\put(19.5,8){1}
\put(17.5,12){2}
\put(31.4,8){1}
\put(33,11.9){1}
\put(23.8,20){1}
\put(27.3,20){4}
\put(16,20){0}
\put(19,16){2}
\put(31.2,16.2){0}
\put(34.5,20){3}
\end{overpic}
\caption{(a)\ a modular directed metric\ \ \ \ \ \ \ \ \ \ \ \ (b)\ the underlying graph of (a)}
\end{center}
\end{figure}For a directed metric $\mu$ on $T$ and $v_0,v_1,\ldots,v_n\in T$, we say that a sequence $(v_0,v_1,\ldots,v_n)$ is $\mu$-\textit{shortest} if $\mu(v_0,v_n)=\sum_{i=0}^{n-1}\mu(v_i,v_{i+1})$. Bandelt \cite{bandelt1985} showed that for a modular (undirected) metric $\mu$, a $\mu$-shortest sequence is also $d_{H_\mu}$-shortest. We have the following directed version of this property:
\begin{lemma}
\label{metric->graph-shortest}
\textit{Let} $\mu$\textit{ be a modular directed metric on} $T$\textit{, and let $v_0,v_1,\ldots,v_n\in T$}.
\begin{enumerate}[label=(\arabic*),ref=\arabic*]
    \item \textit{If a sequence} $(v_0,v_1,\ldots,v_n)$ \textit{is} $\mu$-\textit{shortest}, \textit{then the inverted sequence} $(v_n,v_{n-1},\ldots,v_0)$ \textit{is also} $\mu$-\textit{shortest}.
    \label{lemma14:case1}
    \item \textit{If a sequence} $(v_0,v_1,\ldots,v_n)$ \textit{is} $\mu$-\textit{shortest, then this sequence is also} $d_{H_\mu}$-\textit{shortest}.
    \label{lemma14:case2}
\end{enumerate}
\end{lemma}
\begin{proof}
In this proof, we denote $I:=I_\mu$. We first show (\ref{lemma14:case1}) by induction on $n$. Suppose that $(v_0,v_1,\ldots,v_n)$ is $\mu$-shortest. If $n=1$ then $(v_1,v_0)$ is $\mu$-shortest and (\ref{lemma14:case1}) holds. Suppose that $n\geq 2$. Since $\mu$ is modular, there exists a median $m\in T$ of $v_0,v_1,v_n$. There exists no element $z\in I(v_0,v_1)\cap I(v_1,v_n)\setminus \{v_1\}$, since otherwise a sequence $(v_0,z,v_1,z,v_n)$ is $\mu$-shortest, which is impossible by $\mu(z,v_1)+\mu(v_1,z)>0$. Hence, we have $m=v_1$. Then a sequence $(v_n,v_1,v_0)$ is $\mu$-shortest. Therefore, it suffices to show that a sequence $(v_n,v_{n-1},\ldots,v_1)$ is $\mu$-shortest, which is implied by induction. Hence, we can conclude that $(v_n,v_{n-1},\ldots,v_0)$ is $\mu$-shortest.

Next we show (\ref{lemma14:case2}). We denote $d:=d_{H_\mu}$ for simplicity. Suppose that $(v_0,v_1,\ldots,v_n)$ is $\mu$-shortest. Without loss of generality, we may assume that $v_i\neq v_j$ for any $i\neq j$. If $n=1$, then $(v_0,v_1)$ is obviously $d$-shortest. Suppose that $n\geq 2$. If $d(v_0,v_n)=1$, then $H_\mu$ has the edge $\{v_0,v_n\}$. However, this contradicts $\mu(v_0,v_n)=\mu(v_0,v_1)+\mu(v_1,v_n)$ and the definition of $H_\mu$ because of (\ref{lemma14:case1}). Hence, it suffices to consider the case when $d(v_0,v_n)\geq 2$. In the case of $n=2$ with $I(v_0,v_1)=\{v_0,v_1\}$ and $I(v_1,v_2)=\{v_1,v_2\}$, the underlying graph $H_\mu$ has the edges $\{v_0,v_1\}$ and $\{v_1,v_2\}$. Then the edge $v_0v_2$ is not contained in $H_\mu$, and hence $(v_0,v_1,v_2)$ is $d$-shortest. In the case of $n=2$ with $I(v_0,v_1)\neq \{v_0,v_1\}$, there exists an element $z\in I(v_0,v_1)\setminus \{v_0,v_1\}$. Then the case of $(v_0,v_1,v_2)$ is reduced to the case of $(v_0,z,v_1,v_2)$. Similarly, the case of $n=2$ with $I(v_1,v_2)\neq \{v_1,v_2\}$ is reduced to the case of $n=3$. 

Thus, it suffices to consider the case of $n\geq 3$. We show by induction on $d(v_0,v_n)$ that $(v_0,v_1,\ldots,v_n)$ is $d$-shortest. Let
$(u_0,u_1,\ldots,u_l)$ be a shortest path in $H_\mu$ from $v_0$ to $v_n$ ($u_0=v_0,u_l=v_n$). Since $\mu$ is modular, there exists a median $m\in T$ of $u_0,u_{l-1},v_{n-1}$. Then we have $m\in I(u_0,u_{l-1})$. We may assume that $(u_0,m,u_{l-1})$ is $d$-shortest by induction hypothesis. Then $(u_0,m,u_l)$ is also $d$-shortest. Since $(u_0,m,v_{n-1})$ is $\mu$-shortest, $(m,v_{n-1},v_n)$ is also $\mu$-shortest. Suppose that $m\neq u_0$. In this case, $(m,v_{n-1},v_n)$ is $d$-shortest by induction hypothesis. This implies that $(u_0,m,v_{n-1},v_n)$ is $d$-shortest, because $(u_0,m,u_l)$ is $d$-shortest. We may also assume that $(v_0,v_1,\ldots,v_{n-1})$ is $d$-shortest by induction hypothesis. Hence, we conclude that  $(v_0,v_1,\ldots,v_{n})$ is $d$-shortest. 

Consider the case when $m=u_0$. In this case, $(u_{l-1},u_0,v_{n-1})$ is $\mu$-shortest. Then we have $\mu(u_{l-1},v_{n-1})=\mu(u_{l-1},v_1)+\mu(v_1,v_{n-1})$. 
We next apply the similar argument to $u_l,u_{l-1},v_1$, which we apply to $u_0,u_{l-1},v_{n-1}$ above. Since $\mu$ is modular, there exists a median $m'\in T$ of $u_l,u_{l-1},v_1$. Then we may assume that $(u_{l-1},m',u_l)$ is $d$-shortest by induction hypothesis. Hence, $(u_0,m',u_l)$ is also $d$-shortest. Since $(v_1,m',u_l)$ is $\mu$-shortest, $(v_0,v_1,m')$ is also $\mu$-shortest. Suppose that $m'\neq u_l$. In this case, $(v_0,v_1,m')$ is $d$-shortest by induction hypothesis, which implies that $(v_0,v_1,m',u_l)$ is $d$-shortest. We may also assume that $(v_1,v_2,\ldots,v_n)$ is $d$-shortest by induction hypothesis. Hence, we conclude that  $(v_0,v_1,\ldots,v_{n})$ is $d$-shortest. Thus, it suffices to consider the case when $m'=u_l$. In this case, since $(u_{l-1},u_l,v_1)$ is $\mu$-shortest, we have $\mu(u_{l-1},v_1)=\mu(u_{l-1},v_{n-1})+\mu(v_{n-1},v_1)$.  
Thus, we have $\mu(v_1,v_{n-1})+\mu(v_{n-1},v_1)=0$ (recall that $\mu(u_{l-1},v_{n-1})=\mu(u_{l-1},v_1)+\mu(v_1,v_{n-1})$). This is a contradiction.
\end{proof}

For a modular directed metric $\mu$ on $T$, let $m$ be a median of $x,y,z\in T$ in $\mu$. Then, by Lemma~\ref{metric->graph-shortest} $m$ is also a median of $x,y,z$ in $H_\mu$. Hence, we have the following lemma:
\begin{lemma}
\label{lemma:mu is modular -> H is modular}
\textit{If a directed metric} $\mu$ \textit{is modular}, \textit{then} $H_\mu$ \textit{is also modular}.
\end{lemma}
\subsection{Directed orbits and directed orbit invariance}
Let $G=(V,E)$ be an undirected graph. Let $\overleftrightarrow{E}:=\{(u,v),\ (v,u)\mid \{u,v\}\in E\}\subseteq V\times V$, and $\overleftrightarrow{G}:=(V,\overleftrightarrow{E})$. An element of $\overleftrightarrow{E}$ is called an \textit{oriented edge} of $E$. For a path $P$ from $s$ to $t$ in $G$, we orient each edge of $P$ along the direction of $P$, and we denote by $\overrightarrow{P}$ the corresponding path in $\overleftrightarrow{G}$. Let $\overrightarrow{P}$ and $\overrightarrow{W}$ be paths in $\overleftrightarrow{G}$ such that the end point of $\overrightarrow{P}$ and the start point of $\overrightarrow{W}$ are identified. Then we denote by $\overrightarrow{P}\cup \overrightarrow{W}$ the path obtained by concatenating $\overrightarrow{P}$ and $\overrightarrow{W}$ in this order. In particular, if $\overrightarrow{W}$ consists of one oriented edge $(p,q)$, then we simply denote $\overrightarrow{W}:=(p,q)$ and $\overrightarrow{P}\cup \overrightarrow{W}:=\overrightarrow{P}\cup (p,q)$. For $\overrightarrow{e},\overrightarrow{e'}\in \overleftrightarrow{E}$, we say that $\overrightarrow{e}$ and $\overrightarrow{e'}$ are \textit{projective} if there exists a sequence $(\overrightarrow{e}= \overrightarrow{e_0}, \overrightarrow{e_1},...,\overrightarrow{e_m}=\overrightarrow{e'})\ (\overrightarrow{e_i}=(p_i,q_i)\in \overleftrightarrow{E}\ \mathrm{for\ each\ }i)$ such that
$(p_i,q_i,q_{i+1},p_{i+1},p_i)$ is a 4-cycle in $G$ for each $i$. An equivalence class of the projectivity relation is called a \textit{directed orbit}. Then we have the following lemma about the number of oriented edges of each directed orbit included in a shortest path. This is a sharpening of the result for undirected graphs due to Bandelt \cite{bandelt1985}, and is similarly shown by the proof of the undirected version.
\begin{lemma}[\cite{bandelt1985}]
\label{lemma:orbiteq}
\textit{Let} $G=(V,E)$ \textit{be a modular graph, and let }$\overrightarrow{Q}$ \textit{be a directed orbit. For} $x,y\in V$, \textit{let} $P$ \textit{be a path from} $x$ \textit{to} $y$, \textit{and let} $P^*$ \textit{be a shortest path from} $x$ \textit{to} $y$. \textit{Then we have}
\begin{align}
\label{orbiteq}
    |\overrightarrow{P^*}\cap \overrightarrow{Q}|\leq |\overrightarrow{P}\cap \overrightarrow{Q}|.
\end{align}
\end{lemma}

Let $G=(V,E)$ be an undirected graph. If a function $h:\overleftrightarrow{E}\rightarrow \mathbb{R}_{+}$ satisfies $h(\overrightarrow{e})=h(\overrightarrow{e'})$ for every $\overrightarrow{e},\overrightarrow{e'}\in \overleftrightarrow{E}$ belonging to the same directed orbit, then we say that $h$ is \textit{directed orbit-invariant}. Let $\mu$ be a directed metric on $T$ with the underlying graph $H_\mu=(T,U)$. We say that $\mu$ is directed orbit-invariant if $\mu(u_1,u_2)=\mu(u_1',u_2')$ holds for every $\overrightarrow{u}=(u_1,u_2),\overrightarrow{u'}=(u_1',u_2')\in \overleftrightarrow{U}$ belonging to the same directed orbit in $H_\mu$. A 4-cycle $(p,q,r,s,p)$ in $H_\mu$ is called a \textit{directed orbit-varying modular cycle} if $\mu(p,q)-\mu(s,r)=\mu(r,s)-\mu(q,p)=\mu(p,s)-\mu(q,r)=\mu(r,q)-\mu(s,p)\neq0$. The cycle $(p,q,r,s,p)$ in Figure 3 (b) is an example of a directed orbit-varying modular cycle. 

Bandelt \cite{bandelt1985} showed that a metric $\mu$ is orbit-invariant if $\mu$ is modular. A directed metric $\mu$ is not necessarily directed orbit-invariant even if $\mu$ is modular. For example, if $H_\mu$ is a directed orbit-varying modular cycle, then $\mu$ is modular but not directed orbit-invariant. The name ``directed orbit-varying modular cycle'' is motivated by this fact. We now have the following sufficient condition of a directed metric to be directed orbit-invariant. 
\begin{lemma}
\label{lem:orbitvariant}
\textit{Let} $\mu$ \textit{be a modular directed metric}. \textit{Suppose that} $H_\mu$ \textit{has no directed orbit-varying modular cycle. Then}, $\mu$ \textit{is directed orbit-invariant}.
\end{lemma}
\begin{proof}
It suffices to show that $\mu(p,q)=\mu(s,r)$ for any 4-cycle $(p,q,r,s,p)$ in $H_\mu$. Suppose to the contrary that a 4-cycle $(p,q,r,s,p)$ in $H_\mu$ satisfies $\mu(p,q)\neq \mu(s,r)$. Let $k:=\mu(p,q)-\mu(s,r)\neq 0$. Since $\mu$ is modular, the triple $p,q,r$ has a median $m$. The underlying graph $H_\mu$ has edges $\{p,q\}$ and $\{q,r\}$, which implies $m=q$. Hence, the sequence $(p,q,r)$ is $\mu$-shortest. Similarly, $(p,s,r)$ is $\mu$-shortest. Then we have $\mu(p,q)+\mu(q,r)=\mu(p,s)+\mu(s,r)=\mu(p,r)$. Hence, we have $\mu(p,s)-\mu(q,r)=\mu(p,q)-\mu(s,r)=k$. Similarly, sequences $(s,p,q)$ and $(s,r,q)$ are $\mu$-shortest, then we have $\mu(s,p)+\mu(p,q)=\mu(s,r)+\mu(r,q)$. Hence, we have $\mu(r,q)-\mu(s,p)=\mu(p,q)-\mu(s,r)=k$. Similarly, sequences $(q,p,s)$ and $(q,r,s)$ are $\mu$-shortest, then we have $\mu(q,p)+\mu(p,s)=\mu(q,r)+\mu(r,s)$. Hence, we have $\mu(r,s)-\mu(q,p)=\mu(p,s)-\mu(q,r)=k$. Therefore, we have $\mu(p,q)-\mu(s,r)=\mu(p,s)-\mu(q,r)=\mu(r,q)-\mu(s,p)=\mu(r,s)-\mu(q,p)=k\neq 0$, then we conclude that the cycle $(p,q,r,s,p)$ is a directed orbit-varying modular cycle. This is a contradiction.
\end{proof}

We now consider a sufficient condition for the converse of Lemma \ref{metric->graph-shortest} (\ref{lemma14:case2}) to hold. For an undirected metric $\mu$, Bandelt \cite{bandelt1985} showed that if $\mu$ is orbit-invariant and $H_\mu$ is modular, then a $d_{H_\mu}$-shortest sequence is also $\mu$-shortest. The similar property also holds for a directed metric as follows:
\begin{lemma}
\label{lem:d-shortest->mu-shortest}
\textit{Let $\mu$ be a directed metric on $T$, and let $v_0,v_1,\ldots,v_n\in T$. If $\mu$ is directed orbit-invariant and $H_\mu$ is modular, then the following condition holds:}
\begin{align}
    \textit{If\ }(v_0,v_1,\ldots,v_n)\ \textit{is}\ 
    d_{H_\mu}\textit{-shortest,}\ \textit{then\ it\ is\ also}\ \mu\textit{-shortest.}
\end{align}
\end{lemma}
\begin{proof}
Without loss of generality, we may assume that $(v_0,v_1,\ldots,v_n)$ is a shortest path from $v_0$ to $v_n$ in $H_\mu$. Also, by the definition of $H_\mu$ we see that there is a path $(v_0=u_0,u_1,\ldots,u_m=v_n)$ in $H_\mu$ with $\mu(u_0,u_m)=\mu(u_0,u_1)+\mu(u_1,u_2)+\cdots+\mu(u_{m-1},u_m)$. Hence, by Lemma \ref{lemma:orbiteq} and directed orbit-invariance of $\mu$, we have
\begin{align}
    \mu(v_0,v_1)+\mu(v_1,v_2)+\cdots+\mu(v_{n-1},v_n)&\leq \mu(u_0,u_1)+\mu(u_1,u_2)+\cdots+\mu(u_{m-1},u_m)\notag\\&=\mu(v_0,v_n).
\end{align}
Thus, we have $\mu(v_0,v_1)+\mu(v_1,v_2)+\cdots+\mu(v_{n-1},v_n)=\mu(v_0,v_n)$.
\end{proof}
As we obtain Lemma \ref{lemma:mu is modular -> H is modular} from Lemma \ref{metric->graph-shortest}, we also obtain the following property from Lemma~\ref{lem:d-shortest->mu-shortest} by applying a similar argument:
\begin{lemma}
\textit{Let $\mu$ be a directed metric. If $\mu$ is directed orbit-invariant and $H_\mu$ is modular, then $\mu$ is also modular.}
\end{lemma}

\section{Proof of tractability}
\label{sec:proof of tractability}
In this section, we give proofs of Theorem \ref{thm:directed-p} and Theorem \ref{thm:directed-p-star}. Let $\mu$ be a directed metric on $T$, and $\Gamma$ be the constraint language defined in (\ref{0-ext-->vcsp}). Then, as we see in Section \ref{subsec:valued csp}, $\overrightarrow{\textbf{0}}$\textbf{-Ext}$[\mu]$ is exactly VCSP$[\Gamma]$. Hence, by Theorem \ref{vcsp} we can prove the tractability of $\overrightarrow{\textbf{0}}$\textbf{-Ext}$[\mu]$ by constructing a fractional polymorphism with a semilattice operation in its support. In the proof of Theorem~\ref{thm:directed-p}, we construct a fractional polymorphism that characterizes submodularity of $\mu$. To show submodularity, we imitate the proof of submodularity of metric functions on modular semilattices in the undirected version \cite{hirai2016}. In the proof of Theorem \ref{thm:directed-p-star}, we also construct fractional polymorphisms that are similar to the fractional polymorphisms characterizing submodularity.

\subsection{Proof of Theorem \ref{thm:directed-p}}
Note that the underlying graph $H_\mu$ of $\mu$ is the covering graph of a modular lattice $\mathcal{L}$ with a partial order $\preceq$. We define a partial order $\preceq$ on $\mathcal{L}\times \mathcal{L}$ by $(a,b)\preceq (c,d)\Longleftrightarrow a\preceq c$ and $b\preceq d\ (a,b,c,d\in \mathcal{L})$. Then $\mathcal{L}\times \mathcal{L}$ is also a modular lattice. If $\mu$ is a submodular function on $\mathcal{L}\times \mathcal{L}$, then by Theorem~\ref{vcsp} we can conclude that $\overrightarrow{\textbf{0}}$\textbf{-Ext}$[\mu]$ is solvable in polynomial time. Hence, the following property completes the proof:
\begin{theorem}
\textit{Let }$\mu$ \textit{be a directed metric. Suppose that }$H_\mu$ \textit{is the covering graph of a modular lattice }$\mathcal{L}$\textit{ and }$\mu$\textit{ is directed orbit-invariant. Then the function }$\mu:\mathcal{L}\times \mathcal{L}\rightarrow \mathbb{R}_+$\textit{ is submodular.}
\end{theorem}
\begin{proof}
Note that $\mathcal{L}\times \mathcal{L}$ is a modular lattice. By Lemma \ref{submodular}, $\mu$ is a submodular function on $\mathcal{L}\times \mathcal{L}$ if and only if $\mu(a)+\mu(b)\geq \mu(a\vee b)+\mu(a\wedge b)$ holds for every 2-covered pair $(a,b)\ (a,b\in \mathcal{L}\times \mathcal{L})$. Thus, it suffices to show that $\mu(a)+\mu(b)\geq \mu(a\vee b)+\mu(a\wedge b)$ holds for any 2-covered pair $(a,b)$. Let $a=(a_1,a_2),b=(b_1,b_2)\ (a_1,a_2,b_1,b_2\in \mathcal{L})$. Then, it suffices to consider  the following two cases:
\begin{enumerate}[label=(\roman*),ref=\roman*]
    \item $a_1=b_1$, and $a_2\vee b_2$ covers $a_2,b_2$.
    \label{case1}
    \item $a_1$ covers $b_1$, and $b_2$ covers $a_2$.
    \label{case2}
\end{enumerate}
We first consider the case (\ref{case1}). It suffices to show that $\mu(a_1,a_2)+\mu(a_1,b_2)\geq \mu(a_1,a_2\vee b_2)+\mu(a_1,a_2\wedge b_2)$. Let $Y:=[a_2\wedge b_2,a_2\vee b_2]$. Then, for every $y\in Y\setminus \{a_2\wedge b_2,a_2\vee b_2\}$, it holds that $a_2\vee b_2$ covers $y$ and $y$ covers $a_2\wedge b_2$, because of Lemma \ref{jordandedekind} and Lemma \ref{lem:modular}. Hence, $Y$ is a convex set in the metric space $(T,d)$ (in this proof, we denote $d:=d_{H_\mu}$ for simplicity). Since $\mathcal{L}$ is modular, by Lemma \ref{modular} $H_\mu$ is modular. Then, by Lemma \ref{gated} $Y$ is a gated set. Hence, there exists $y^*\in Y$ such that $d(a_1,y)=d(a_1,y^*)+d(y^*,y)$ holds for every $y\in Y$. Therefore, by Lemma \ref{lem:d-shortest->mu-shortest}, $(a_1,y^*,y)$ is $\mu$-shortest for every $y\in Y$. If $y^*=a_2$, then we have
\begin{align}
    \mu(a_1,a_2\vee b_2)&=\mu(a_1,a_2)+\mu(a_2,a_2\vee b_2),\notag\\
    \mu(a_1,a_2\wedge b_2)&=\mu(a_1,a_2)+\mu(a_2,a_2\wedge b_2), \notag \\
    \mu(a_1,b_2)&=\mu(a_1,a_2)+\mu(a_2,b_2).
\end{align}
Furthermore, since $\mu$ is directed orbit-invariant, we have $\mu(a_2\vee b_2,b_2)=\mu(a_2,a_2\wedge b_2)$. In addition, by Lemma \ref{lem:d-shortest->mu-shortest} we have $\mu(a_2,b_2)=\mu(a_2,a_2\vee b_2)+\mu(a_2\vee b_2,b_2)$. Hence, we have $\mu(a_1,b_2)=\mu(a_1,a_2)+\mu(a_2,a_2\vee b_2)+\mu(a_2,a_2\wedge b_2)$. Therefore, we obtain $\mu(a_1,a_2)+\mu(a_1,b_2)=\mu(a_1,a_2\vee b_2)+\mu(a_1,a_2\wedge b_2)$. Similarly, if $y^*=b_2$, we obtain $\mu(a_1,a_2)+\mu(a_1,b_2)=\mu(a_1,a_2\vee b_2)+\mu(a_1,a_2\wedge b_2)$. If $y^*=a_2\vee b_2$, then we have
\begin{align}
\label{equation:cube}
    \mu(a_1,a_2)&=\mu(a_1,a_2\vee b_2)+\mu(a_2\vee b_2,a_2),\notag\\
    \mu(a_1,b_2)&=\mu(a_1,a_2\vee b_2)+\mu(a_2\vee b_2,b_2), \notag \\
    \mu(a_1,a_2\wedge b_2)&=\mu(a_1,a_2\vee b_2)+\mu(a_2\vee b_2,a_2)+\mu(a_2,a_2\wedge b_2).
\end{align}
Since $\mu$ is directed orbit-invariant, we have $\mu(a_2\vee b_2,b_2)=\mu(a_2,a_2\wedge b_2)$. Hence, by (\ref{equation:cube}) we obtain $\mu(a_1,a_2)+\mu(a_1,b_2)=\mu(a_1,a_2\vee b_2)+\mu(a_1,a_2\wedge b_2)$. Similarly, if $y^*=a_2\wedge b_2$, then we obtain $\mu(a_1,a_2)+\mu(a_1,b_2)=\mu(a_1,a_2\vee b_2)+\mu(a_1,a_2\wedge b_2)$. Thus, it suffices to consider the case when $y^*\neq a_2,b_2,a_2\vee b_2,a_2\wedge b_2$. In this case, we have
\begin{align}
    \mu(a_1,a_2)&=\mu(a_1,y^*)+\mu(y^*,a_2\vee b_2)+\mu(a_2\vee b_2,a_2)\geq \mu(a_1,a_2\vee b_2),\notag\\ 
    \mu(a_1,b_2)&=\mu(a_1,y^*)+\mu(y^*,a_2\wedge b_2)+\mu(a_2\wedge b_2,b_2)\geq \mu(a_1,a_2\wedge b_2).
\end{align}
Hence, we have $\mu(a_1,a_2)+\mu(a_1,b_2)\geq \mu(a_1,a_2\vee b_2)+\mu(a_1,a_2\wedge b_2)$.

For the next, we consider the case (\ref{case2}). The submodularity is $\mu(a_1,a_2)+\mu(b_1,b_2)\geq \mu(a_1,b_2)+\mu(b_1,a_2)$. Since $H_\mu$ is bipartite, $d(a_1,b_2)$ is equal to  either $d(b_1,a_2)$ or $d(b_1,a_2)+2$ or $d(b_1,a_2)-2$. If $d(a_1,b_2)$ is equal to $d(b_1,a_2)+2$ or $d(b_1,a_2)-2$, then by Lemma \ref{lem:d-shortest->mu-shortest} we have $\mu(a_1,a_2)+\mu(b_1,b_2)= \mu(a_1,b_2)+\mu(b_1,a_2)$. Thus, it suffices to consider the case when $d(a_1,b_2)=d(b_1,a_2)$. In this case, $d(a_1,a_2)$ is equal to either $d(a_1,b_2)-1$ or $d(a_1,b_2)+1$. Suppose that $d(a_1,a_2)=d(a_1,b_2)+1$. Then, by Lemma \ref{lem:d-shortest->mu-shortest} we have $\mu(a_1,a_2)=\mu(a_1,b_2)+\mu(b_2,a_2)$. Hence, we obtain $\mu(a_1,a_2)+\mu(b_1,b_2)=\mu(a_1,b_2)+\mu(b_2,a_2)+\mu(b_1,b_2)\geq \mu(a_1,b_2)+\mu(b_1,a_2)$. Consider the case when $d(a_1,a_2)=d(a_1,b_2)-1$. Similarly, $d(b_1,b_2)$ is equal to either $d(a_1,b_2)-1$ or $d(a_1,b_2)+1$, and by the similar argument, we may assume that $d(b_1,b_2)=d(a_1,b_2)-1$. Let $P$ be a shortest path in $H_\mu$ from $a_1$ to $a_2$. Let $z$ be the vertex in $P$ that is adjacent to $a_1$. Then, we have $d(z,b_2)=d(b_1,b_2)=d(a_1,b_2)-1$. Hence, by Lemma \ref{modularcondition}, there exists a common neighbor $w$ of $z,b_1$ with $d(w,b_2)=d(a_1,b_2)-2$. Then, we have $d(z,b_2)=d(w,a_2)=d(a_1,b_2)-1,\ d(z,a_2)=d(z,b_2)-1,$ and $d(w,b_2)=d(z,b_2)-1$. Furthermore, since $a_1$ covers $b_1$, we see that $z$ covers $w$. Hence, we can apply the same argument to $z,w,a_2,b_2$ which we apply to $a_1,b_1,a_2,b_2$ above. 
By repeating this argument, we can see that $a_2$ covers $b_2$, but this is a contradiction.
\end{proof}

\subsection{Proof of Theorem \ref{thm:directed-p-star}}
To prove the tractability, we construct fractional polymorphisms which satisfy the property of Theorem \ref{vcsp}. Let $r$ denote the internal node of a star $H_\mu$. A subset $X\subseteq T\setminus \{r\}$ is called \textit{unbiased} if $R_{\mu}(p,r)=R_{\mu}(r,q)$ holds for every $p,q\in X\ (p\neq q)$. Note that if $X$ is unbiased and $|X|\geq 3$, then $\mu(p,r)=\mu(r,p)$ holds for any $p\in X$, since $R_\mu(p,r)=R_\mu(r,q)=R_\mu(s,r)=R_\mu(r,p)$ holds for $q,s\in X\setminus \{p\}\ (q\neq s)$. We divide $T\setminus \{r\}$ into the minimum number of disjoint unbiased sets, and denote by $\mathcal{F}$ the family of them. From the assumption that there exists no biased non-collinear triple, the family $\mathcal{F}$ consists of at most two sets. Hence, it suffices to consider the following three cases:
\begin{enumerate}[label=(\roman*),ref=\roman*]
    \item $\mathcal{F}=\{X,Y\}$, and $|X|, |Y|\leq 2$.
    \label{thm6:case1}
    \item $\mathcal{F}=\{X,Y\}$, and $|X|\geq 3,\ |Y|\leq 2$.
    \label{thm6:case2}
    \item $\mathcal{F}=\{X\}$.
    \label{thm6:case3}
\end{enumerate}

We first consider the case (\ref{thm6:case1}). It suffices to consider the case when $|X|=|Y|=2$, since fractional polymorphisms for this case work for the other cases. Let $X=\{x_1,x_2\}$ and $Y=\{y_1,y_2\}$. Let $\preceq$ be a partial order on $T$ defined by $y_i\prec r\prec x_j\ (i=1,2,\ j=1,2)$ (see Figure 4 (a)). For $i=1,2$, we extend $\preceq$ to $\preceq_i$ by adding relation $y_i \prec_i y_j\ (j\neq i)$. Then each pair of two elements $t_1,t_2$ in a partially ordered set $(T,\preceq_i)$ has a unique meet, denoted by $t_1\wedge_i t_2$ ($i=1,2$).
\begin{figure}[tbp]
\begin{center}
\begin{overpic}[width=14cm]{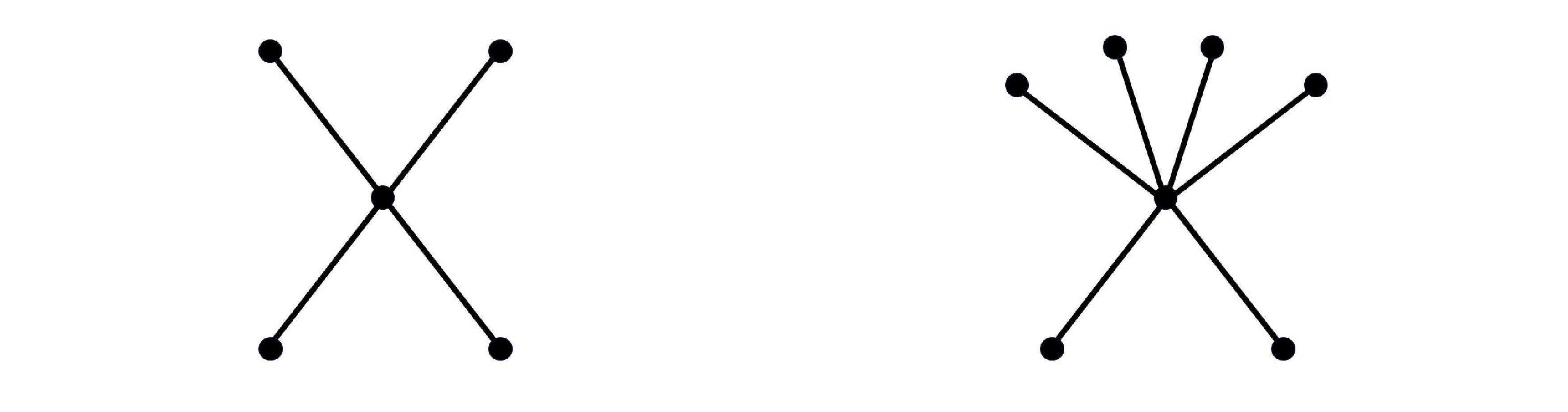}
\put(13,23){$x_1$}
\put(33.7,23){$x_2$}
\put(26,13){$r$}
\put(13,4){$y_1$}
\put(33.7,4){$y_2$}
\put(60.5,20.3){$x_1$}
\put(85.7,20.3){$x_4$}
\put(68,25){$x_2$}
\put(77.7,24.7){$x_3$}
\put(76.5,13){$r$}
\put(63,4){$y_1$}
\put(83.6,4){$y_2$}
\end{overpic}
\caption{(a)\ the case (i)\ \ \ \ \ \ \ \ \ \ \ \ \ \ \ \ \ \ \ \ \ \ (b)\ the case (ii)\ ($k=4$)}
\end{center}
\end{figure}
For $i=1,2$, similar to $\preceq_i$, we also extend $\preceq$ to $\preceq_i'$ by adding relation $x_j \prec_i' x_i\ (j\neq i)$. Then each pair of two elements $t_1,t_2$ in $(T,\preceq_i')$ has a unique join, denoted by $t_1\vee_i t_2$ ($i=1,2$).
Since $X$ and $Y$ are unbiased, the lengths of the edges in $H_\mu$ can be written as
\begin{align}
    &\mu(x_1,r)=a,\ \mu(r,x_1)=b,\ \mu(r,x_2)=ka,\ \mu(x_2,r)=kb,\notag\\
    &\mu(r,y_1)=c,\ \mu(y_1,r)=d,\ \mu(y_2,r)=lc,\ \mu(r,y_2)=ld.
\end{align}
Since $\wedge_i,\vee_j$ are semilattice operations, by Theorem \ref{vcsp} it suffices to show that for any $t^1,t^2\in T\times T$,
\begin{align}
\label{vcspcondition:case1}
    \mu(t^1)+\mu(t^2)\geq \frac{1}{l+1} \mu(t^1\wedge_1 t^2)+\frac{l}{l+1} \mu(t^1\wedge_2 t^2)+\frac{1}{k+1} \mu(t^1\vee_1 t^2)+\frac{k}{k+1} \mu(t^1\vee_2 t^2).
\end{align}
Let $t^1=(t_1^1,t_2^1)$ and $t^2=(t_1^2,t_2^2)$. 
If a pair $t_i^1,t_i^2$ has both meet $t_i^1\wedge t_i^2$ and join $t_i^1\vee t_i^2$ in $(T,\preceq)$ for all $i=1,2$, then (\ref{vcspcondition:case1}) reduces to Theorem \ref{thm:directed-p}. 
Suppose that a pair $t_i^1,t_i^2$ has no meet or has no join for some $i$ in $(T,\preceq)$. Without loss of generality, we may assume that $t_1^1=x_1$ and $t_1^2=x_2$. If $t_2^1=t_2^2=y_i$ ($i\in \{1,2\}$), then we have
\begin{align}
    &\mu(x_1,y_i)+\mu(x_2,y_i)= \mu(x_1,r)+\mu(r,y_i)+\mu(x_2,r)+\mu(r,y_i)\notag\\
    &\geq \frac{1}{k+1} \mu(x_1,r)+\frac{1}{k+1} \mu(r,y_i)+\frac{k}{k+1} \mu(r,y_i)+\frac{k}{k+1} \mu(x_2,r)+\frac{1}{l+1} \mu(r,y_i)+\frac{l}{l+1} \mu(r,y_i)\notag\\
    &=\frac{1}{l+1} \mu(x_1\wedge_1 x_2,y_i\wedge_1 y_i)+\frac{l}{l+1} \mu(x_1\wedge_2 x_2,y_i\wedge_2 y_i)+\frac{1}{k+1} \mu(x_1\vee_1 x_2,y_i\vee_1 y_i)\notag\\
    &+\frac{k}{k+1} \mu(x_1\vee_2 x_2,y_i\vee_2 y_i).
\end{align}
Thus, (\ref{vcspcondition:case1}) holds. Suppose that $t_2^1,t_2^2\in \{r,y_1,y_2\}$ and $(t_2^1,t_2^2)\neq (y_i,y_i)$ for $i=1,2$. Then we have $t_2^1\wedge_i t_2^2\in \{t_2^1,t_2^2\}$ and $t_2^1\vee_i t_2^2=r$ for $i=1,2$. Hence, we have
\begin{align}
    &\mu(x_1,t_2^1)+\mu(x_2,t_2^2)= \mu(x_1,r)+\mu(x_2,r)+\mu(r,t_2^1)+\mu(r,t_2^2)\notag\\
    &\geq \frac{1}{k+1} \mu(x_1,r)+\frac{k}{k+1} \mu(x_2,r)+\frac{1}{l+1} \mu(r,t_2^1\wedge_1 t_2^2)+\frac{l}{l+1} \mu(r,t_2^1\wedge_2 t_2^2)\notag\\
    &= \frac{1}{l+1} \mu(x_1\wedge_1 x_2,t_2^1\wedge_1 t_2^2)+\frac{l}{l+1} \mu(x_1\wedge_2 x_2,t_2^1\wedge_2 t_2^2)+\frac{1}{k+1} \mu(x_1\vee_1 x_2,t_2^1\vee_1 t_2^2)\notag\\
    &+\frac{k}{k+1} \mu(x_1\vee_2 x_2,t_2^1\vee_2 t_2^2).
\end{align}
Thus, it suffices to consider the case when $t_2^1=x_1$ or $x_2$. We first consider the case when $t_2^1=x_1$. If $t_2^2=x_2$, then (\ref{vcspcondition:case1}) holds trivially, since the both sides of (\ref{vcspcondition:case1}) are 0. If $t_2^2\in \{x_1,r,y_1,y_2\}$, then we have
\begin{align}
    &\mu(x_1,x_1)+\mu(x_2,t_2^2)=\mu(x_2,r)+\mu(r,t_2^2)\notag\\
    &=\frac{k}{k+1}\mu(x_2,x_1)+\frac{1}{l+1}\mu(r,t_2^2)+\frac{l}{l+1}\mu(r,t_2^2) \notag\\
    &=\frac{1}{l+1} \mu(x_1\wedge_1 x_2,x_1\wedge_1 t_2^2)+\frac{l}{l+1} \mu(x_1\wedge_2 x_2,x_1\wedge_2 t_2^2)+\frac{1}{k+1} \mu(x_1\vee_1 x_2,x_1\vee_1 t_2^2)\notag\\
    &+\frac{k}{k+1} \mu(x_1\vee_2 x_2,x_1\vee_2 t_2^2),
\end{align}
where we use $\mu(x_2,r)=kb=\frac{k}{k+1}\mu(x_2,x_1)$ and $\mu(x_1\vee_1 x_2,x_1\vee_1 t_2^2)=\mu(x_1,x_1)=0$.
Thus, (\ref{vcspcondition:case1}) holds. We next consider the case when $t_2^1=x_2$. If $t_2^2=x_1$, then (\ref{vcspcondition:case1}) holds trivially, since the right hand side of (\ref{vcspcondition:case1}) is 0. If $t_2^2\in \{r,y_1,y_2\}$, then we have
\begin{align}
    &\mu(x_1,x_2)+\mu(x_2,t_2^2)=\mu(x_1,x_2)+\mu(x_2,r)+\mu(r,t_2^2)\notag\\
    &\geq \frac{1}{k+1}\mu(x_1,x_2)+\frac{1}{l+1} \mu(r,t_2^2)+\frac{l}{l+1} \mu(r,t_2^2)\notag\\
    &=\frac{1}{l+1} \mu(x_1\wedge_1 x_2,x_2\wedge_1 t_2^2)+\frac{l}{l+1} \mu(x_1\wedge_2 x_2,x_2\wedge_2 t_2^2)+\frac{1}{k+1} \mu(x_1\vee_1 x_2,x_2\vee_1 t_2^2)\notag\\
    &+\frac{k}{k+1} \mu(x_1\vee_2 x_2,x_2\vee_2 t_2^2).
\end{align}
Thus, (\ref{vcspcondition:case1}) holds. This completes the proof of the case (\ref{thm6:case1}).

We next consider the case (\ref{thm6:case2}). Let $X=\{x_1,x_2,\ldots,x_k\}$ and $Y=\{y_1,y_2\}$ ($k\geq 3$). Let $\preceq$ be a partial order on $T$ defined by $y_i\prec r\prec x_j\ (i=1,2,\ j=1,2,\ldots k)$ (see Figure 4 (b)). Similar to the case (\ref{thm6:case1}), for $i=1,2$, we extend $\preceq$ to $\preceq_i$ by adding relation $y_i \prec_i y_j\ (j\neq i)$. Then each pair of two elements $t_1,t_2$ in $(T,\preceq_i)$ has a unique meet, denoted by $t_1\wedge_i t_2$ ($i=1,2$). Let $[k]:=\{1,2,\ldots,k\}$. Let $\mathcal{G}$ be the set of all functions $g:{\scriptsize\begin{pmatrix}[k] \\
                     2 \\
      \end{pmatrix}}\rightarrow [k]$
that satisfy $g(ij)\in \{i,j\}$ for any $ij\in {\scriptsize\begin{pmatrix}[k] \\
                     2 \\
      \end{pmatrix}}$. 
Then, for $g\in \mathcal{G}$, let $\vee_g$ be a binary operation on $T$ defined by $x_i\vee_g x_j:=x_{g(ij)}$ for $ij\in {\scriptsize\begin{pmatrix}[k] \\
                     2 \\
\end{pmatrix}}$, and $t_1\vee_g t_2:=t_1\vee t_2$ for $t_{1}t_2\in {\scriptsize\begin{pmatrix}T \\
                     2 \\
      \end{pmatrix}}\setminus {\scriptsize\begin{pmatrix}X \\
                     2 \\
\end{pmatrix}}$, where $t_1\vee t_2$ is a unique join in $(T,\preceq)$. By $|X|\geq 3$, it holds that $\mu(x_i,r)=\mu(r,x_i)$ for any $i\in [k]$. Since $Y$ is unbiased, the lengths of the edges in $H_\mu$ can be written as
\begin{align}
    &\mu(x_i,r)=\mu(r,x_i)=a_i\ (i\in [k]),\notag\\
    &\mu(r,y_1)=b,\ \mu(y_1,r)=c,\ \mu(y_2,r)=lb,\ \mu(r,y_2)=lc.
\end{align}
Since $\wedge_i$ is a semilattice operation, by Theorem \ref{vcsp} it suffices to show that for any $t^1,t^2\in T\times T$,
\begin{align}
\label{vcspcondition:case2}
    \mu(t^1)+\mu(t^2)\geq \frac{1}{l+1}\mu(t^1\wedge_1 t^2)+\frac{l}{l+1}\mu(t^1\wedge_2 t^2)+\sum_{g\in \mathcal{G}}\prod_{\tiny ij\in
\begin{pmatrix}
[k] \\
2 \\
\end{pmatrix}\normalsize}\frac{a_{g(ij)}}{a_i+a_j}\mu(t^1\vee_g t^2).
\end{align}
Note that $\sum_{g\in \mathcal{G}}\prod_{\tiny ij\in
\begin{pmatrix}
[k] \\
2 \\
\end{pmatrix}\normalsize}a_{g(ij)}=\prod_{\tiny ij\in
\begin{pmatrix}
[k] \\
2 \\
\end{pmatrix}\normalsize}(a_i+a_j)$. Let $t^1=(t_1^1,t_2^1)$ and $t^2=(t_1^2,t_2^2)$. If $|\{t_1^1,t_2^1,t_1^2,t_2^2\}\cap X|\leq 2$, then (\ref{vcspcondition:case2}) reduces to the case (\ref{thm6:case1}). Consider the case when $|\{t_1^1,t_2^1,t_1^2,t_2^2\}\cap X|\geq 3$. Without loss of generality, we may assume that $t_1^1=x_1,t_2^1=x_2,t_1^2=x_3$. If $t_2^2=x_i$ ($i\neq 2,3$), then we have
\begin{align}
    &\mu(x_1,x_2)+\mu(x_3,x_i)=\mu(x_1,r)+\mu(r,x_2)+\mu(x_3,r)+\mu(r,x_i)\notag\\
    &\geq \frac{a_1}{a_1+a_3}\cdot \frac{a_2}{a_2+a_i}\mu(x_1,r)+\frac{a_1}{a_1+a_3}\cdot \frac{a_i}{a_2+a_i}\mu(x_1,r)+\frac{a_1}{a_1+a_3}\cdot \frac{a_2}{a_2+a_i}\mu(r,x_2)\notag\\
    &+\frac{a_3}{a_1+a_3}\cdot \frac{a_2}{a_2+a_i}\mu(r,x_2)+\frac{a_3}{a_1+a_3}\cdot \frac{a_2}{a_2+a_i}\mu(x_3,r)+\frac{a_3}{a_1+a_3}\cdot \frac{a_i}{a_2+a_i}\mu(x_3,r)\notag\\
    &+\frac{a_1}{a_1+a_3}\cdot \frac{a_i}{a_2+a_i}\mu(r,x_i)+\frac{a_3}{a_1+a_3}\cdot \frac{a_i}{a_2+a_i}\mu(r,x_i)\notag\\
    &\geq \frac{a_1}{a_1+a_3}\cdot \frac{a_2}{a_2+a_i}\mu(x_1,x_2)+\frac{a_1}{a_1+a_3}\cdot \frac{a_i}{a_2+a_i}\mu(x_1,x_i)+\frac{a_3}{a_1+a_3}\cdot \frac{a_2}{a_2+a_i}\mu(x_3,x_2)\notag\\
    &+\frac{a_3}{a_1+a_3}\cdot \frac{a_i}{a_2+a_i}\mu(x_3,x_i)\notag\\
    &= \frac{1}{l+1}\mu(x_1\wedge_1 x_3,x_2\wedge_1 x_i)+\frac{l}{l+1}\mu(x_1\wedge_2 x_3,x_2\wedge_2 x_i)\notag\\
    &+\sum_{g\in \mathcal{G}}\prod_{\tiny jm\in
\begin{pmatrix}
[k] \\
2 \\
\end{pmatrix}\normalsize}\frac{a_{g(jm)}}{a_j+a_m}\mu(x_1\vee_g x_3,x_2\vee_g x_i).
\end{align}
(The second inequality is strict when $x_i=x_1$.)
If $t_2^2=x_2$, then we have
\begin{align}
    &\mu(x_1,x_2)+\mu(x_3,x_2)=\mu(x_1,r)+\mu(r,x_2)+\mu(x_3,r)+\mu(r,x_2)\notag\\
    &\geq \frac{a_1}{a_1+a_3}\mu(x_1,r)+\frac{a_1}{a_1+a_3}\mu(r,x_2)+\frac{a_3}{a_1+a_3}\mu(r,x_2)+\frac{a_3}{a_1+a_3}\mu(x_3,r)+\frac{1}{l+1}\mu(r,x_2)\notag\\
    &+\frac{l}{l+1}\mu(r,x_2)\notag\\
    &=\frac{a_1}{a_1+a_3}\mu(x_1,x_2)+\frac{a_3}{a_1+a_3}\mu(x_3,x_2)+\frac{1}{l+1}\mu(r,x_2)+\frac{l}{l+1}\mu(r,x_2)\notag\\
    &= \frac{1}{l+1}\mu(x_1\wedge_1 x_3,x_2\wedge_1 x_2)+\frac{l}{l+1}\mu(x_1\wedge_2 x_3,x_2\wedge_2 x_2)\notag\\
    &+\sum_{g\in \mathcal{G}}\prod_{\tiny jm\in
\begin{pmatrix}
[k] \\
2 \\
\end{pmatrix}\normalsize}\frac{a_{g(jm)}}{a_j+a_m}\mu(x_1\vee_g x_3,x_2\vee_g x_2).
\end{align}
If $t_2^2=x_3$, then we have
\begin{align}
    &\mu(x_1,x_2)+\mu(x_3,x_3)=a_1+a_2\notag\\
    &=\frac{a_1a_2(a_1+a_2)+a_2a_3(a_2+a_3)+a_3a_1(a_3+a_1)+2a_1a_2a_3}{(a_1+a_3)(a_2+a_3)}\notag\\
    &\geq \frac{a_1a_2(a_1+a_2)+a_2a_3(a_2+a_3)+a_3a_1(a_3+a_1)}{(a_1+a_3)(a_2+a_3)}\notag\\
    &=\frac{a_1}{a_1+a_3}\cdot \frac{a_2}{a_2+a_3} \mu(x_1,x_2)+\frac{a_3}{a_1+a_3}\cdot \frac{a_2}{a_2+a_3} \mu(x_3,x_2)+\frac{a_1}{a_1+a_3}\cdot \frac{a_3}{a_2+a_3} \mu(x_1,x_3)\notag\\
    &=\frac{1}{l+1}\mu(x_1\wedge_1 x_3,x_2\wedge_1 x_3)+\frac{l}{l+1}\mu(x_1\wedge_2 x_3,x_2\wedge_2 x_3)\notag\\
    &+\sum_{g\in \mathcal{G}}\prod_{\tiny jm\in
\begin{pmatrix}
[k] \\
2 \\
\end{pmatrix}\normalsize}\frac{a_{g(jm)}}{a_j+a_m}\mu(x_1\vee_g x_3,x_2\vee_g x_3).
\end{align}
If $t_2^2\in \{r,y_1,y_2\}$, then we have
\begin{align}
    &\mu(x_1,x_2)+\mu(x_3,t_2^2)=\mu(x_1,r)+\mu(r,x_2)+\mu(x_3,r)+\mu(r,t_2^2)\notag\\
    &=\mu(x_1,r)+\frac{a_1}{a_1+a_3}\mu(r,x_2)+\frac{a_3}{a_1+a_3}\mu(r,x_2)+\mu(x_3,r)+\mu(r,t_2^2)\notag\\
    &\geq \frac{a_1}{a_1+a_3} \mu(x_1,x_2)+\frac{a_3}{a_1+a_3}\mu(x_3,x_2)+\frac{1}{l+1}\mu(r,t_2^2)+\frac{l}{l+1}\mu(r,t_2^2)\notag\\
    &=\frac{1}{l+1}\mu(x_1\wedge_1 x_3,x_2\wedge_1 t_2^2)+\frac{l}{l+1}\mu(x_1\wedge_2 x_3,x_2\wedge_2 t_2^2)\notag\\
    &+\sum_{g\in \mathcal{G}}\prod_{\tiny jm\in
\begin{pmatrix}
[k] \\
2 \\
\end{pmatrix}\normalsize}\frac{a_{g(jm)}}{a_j+a_m}\mu(x_1\vee_g x_3,x_2\vee_g t_2^2).
\end{align}
This completes the proof of the case (\ref{thm6:case2}).

We finally consider the case (\ref{thm6:case3}). Let $\omega$ be the fractional polymorphism constructed in the case (\ref{thm6:case2}), where $\mathcal{F}=\{X',Y'\}$ and $|X'|\geq 3,|Y'|\leq 2$. For any operation $\varphi$ in $\mathrm{supp}(\omega)$ and any $x_1,x_2\in X',\ y_1,y_2\in Y'$, we have
\begin{align}
    &\varphi(x_1,r),\varphi(r,x_1),\varphi(x_1,x_2)\in X'\cup \{r\},\notag\\
    &\varphi(y_1,r),\varphi(r,y_1),\varphi(y_1,y_2)\in Y'\cup \{r\}.\notag
\end{align}
Hence, $\omega$ also works for the case of $\mathcal{F}=\{X'\}$ or $\mathcal{F}=\{Y'\}$, which completes the proof.

\section{Proof of hardness}
\label{sec:proof of hardness}
In this section, we give proofs of Theorem \ref{thm:directed-nph-extend} and Theorem \ref{thm:directed-nph-new}. We prove them with the aid of Proposition \ref{prop:mcNP-hard}; for a directed metric $\mu$ satisfying the condition in Theorem \ref{thm:directed-nph-extend} or Theorem \ref{thm:directed-nph-new}, and the corresponding constraint language $\Gamma$ defined in (\ref{0-ext-->vcsp}), we show that $\Gamma$ satisfies the condition (MC). To prove this, we construct a ``gadget'' which is a counterexample to submodularity of the objective function of $\overrightarrow{\textbf{0}}$\textbf{-Ext}$[\mu]$ (in a certain sense). This type of hardness proof using a gadget and the condition (MC) originates from the hardness proof of the multiterminal cut problem \cite{dahlhaus1994}. Karzanov \cite{karzanov1998, karzanov2004} also adopted the similar approach to show the hardness results of \textbf{0-Ext}$[\mu]$ (Theorem \ref{thm:undirected-nph}). We apply this type of hardness proof to the directed version $\overrightarrow{\textbf{0}}$\textbf{-Ext}$[\mu]$. We first describe the sufficient condition for $\Gamma$ defined in (\ref{0-ext-->vcsp}) to satisfy the (MC) condition in Section \ref{hardness:approach}. For the next, we prove NP-hardness of $\overrightarrow{\textbf{0}}$\textbf{-Ext}$[\mu]$ by using this sufficient condition in Section \ref{hardness:proofs}.

\subsection{Approach}
\label{hardness:approach}
Let $\mu$ be a rational-valued directed metric on $T$. Suppose that we are given $V\supseteq T$ and $c:V\times V\rightarrow \mathbb{Q}_+$ as an instance of $\overrightarrow{\textbf{0}}$\textbf{-Ext}$[\mu]$. For $s_0,s_1,\ldots,s_k\in T$ and $x_0,x_1,\ldots,x_k\in V\setminus T$, we denote by $\tau_{c} (s_0,x_0|s_1,x_1|\cdots|s_k,x_k)$ the optimal value of $\overrightarrow{\textbf{0}}$\textbf{-Ext}$[\mu]$ subject to $\gamma(x_0)=s_0,\gamma(x_1)=s_1,\ldots,\gamma(x_k)=s_k$. We simply denote $\tau(s_0,x_0|s_1,x_1|\cdots|s_k,x_k):=\tau_{c} (s_0,x_0|s_1,x_1|\cdots|s_k,x_k)$ if $c$ is clear in the context. Let $\tau^*$ be the optimal value of $\overrightarrow{\textbf{0}}$\textbf{-Ext}$[\mu]$. Similar to the constructions in \cite{dahlhaus1994, karzanov1998, karzanov2004}, we desire a pair (called \textit{gadget}) $(V,c)$ satisfying the following properties (in other words, ``violates submodularity,'' cf. \cite{dahlhaus1994}) for specified elements $s,t\in T$ and $x,y\in V\setminus T$.
\begin{align}
\label{nph-condition}
\mathrm{(i)}& \mathrm{\ \ \ }\tau(s,x|t,y)=\tau(t,x|s,y)=\tau^*,\notag\\
\mathrm{(ii)}& \mathrm{\ \ \ }\tau(s,x|s,y)=\tau(t,x|t,y)=\tau^*+\delta \mathrm{\ \ \ }\mathrm{for\ some\ \delta>0},\\
\mathrm{(iii)}& \mathrm{\ \ \ }\tau(s',x|t',y)\geq \tau^*+\delta \mathrm{\ \ \ }\mathrm{for\ all\ other\ pairs}\ (s',t')\in T\times T.\notag
\end{align}
Let $\Gamma$ be a constraint language defined in (\ref{0-ext-->vcsp}). We show that the existence of a gadget $(V,c)$ satisfying (\ref{nph-condition}) implies NP-hardness of VCSP$[\Gamma]$. Suppose that there exists a gadget $(V,c)$ satisfying (\ref{nph-condition}) for $s,t\in T$ and $x,y\in V\setminus T$. Then, we define a function $f:T\times T\rightarrow \mathbb{Q}_+$ as follows:
\begin{align}
    f(t_1,t_2): =\tau(t_1,x|t_2,y)\ \ \ (t_1,t_2\in T).
\end{align}
We have $f\in \langle \Gamma \rangle$ and $\mathrm{argmin}\ f=\{(s,t),(t,s)\}$. Hence, $\Gamma$ satisfies the condition (MC), which implies that VCSP$[\Gamma]$ is strongly NP-hard by Proposition \ref{prop:mcNP-hard}. 

\subsection{Proofs}
\label{hardness:proofs}
In this section, we show Theorem \ref{thm:directed-nph-extend} and Theorem \ref{thm:directed-nph-new} by constructing gadgets $(V,c)$ satisfying (\ref{nph-condition}). We can state Theorem \ref{thm:directed-nph-extend} in the following equivalent form:
\begin{theorem}
\label{thm:directed-nph-extend2}
\textit{Let} $\mu$ \textit{be a rational-valued directed metric.} $\overrightarrow{\textbf{0}}$\textbf{-Ext}$[\mu]$ \textit{is strongly NP-hard if one of the following conditions holds:} 
\begin{enumerate}[label=(\roman*),ref=\roman*]
    \item $\mu$ \textit{is not modular.}
    \label{thm:directed-nph-extend2-case1}
    \item $\mu$ \textit{is modular and not directed orbit-invariant.}
    \label{thm:directed-nph-extend2-case2}
    \item $\mu$ \textit{is modular and directed orbit-invariant, and} $H_\mu$ \textit{is not orientable.}
    \label{thm:directed-nph-extend2-case3}
\end{enumerate}
\end{theorem}
We prove each case of Theorem \ref{thm:directed-nph-extend2} and Theorem \ref{thm:directed-nph-new}.

\subsubsection{Proof of Theorem \ref{thm:directed-nph-extend2} for the case (\ref{thm:directed-nph-extend2-case1})}
We first show the following lemma, which originates from the proof of Theorem \ref{thm:undirected-nph} (\ref{dichotomy:not modular}) in \cite{karzanov2004}.
\begin{lemma}
\label{lem:nphcondition2}
\textit{Let} $\mu$ \textit{be a rational-valued directed metric on }$T$. \textit{If there exists a pair} $(V,c)$ \textit{which satisfies the following properties for a non-collinear triple} $(s_0,s_1,s_2)$ \textit{in} $T$ \textit{and distinct elements} $z_0,z_1,\ldots,z_5\in V\setminus T$\textit{, then} $\overrightarrow{\textbf{0}}$\textbf{-Ext}$[\mu]$ \textit{is strongly NP-hard.}
\begin{align}
\label{nphcondition2}
\mathrm{(i)}& \mathrm{\ \ \ }\tau(s_{i_0+1},z_0|s_{i_1-1},z_1|s_{i_2},z_2|s_{i_3+1},z_3|s_{i_4-1},z_4|s_{i_5},z_5)=\tau^*\mathrm{\ \ \ }(i_j\in \{0,1\}\mathrm{\ \textit{for\ each}\ }j),\notag\\
\mathrm{(ii)}& \mathrm{\ \ \ }\tau(s'_0,z_0|s'_1,z_1|s'_2,z_2|s'_3,z_3|s'_4,z_4|s'_5,z_5)\geq \tau^*+\delta \notag \\&\mathrm{\ \ \ }\mathrm{\textit{for\ all\ other\ sextuplets\ }}s_0',s_1',s_2',s_3',s_4',s_5'\in T\mathrm{\textit{\ and\ some}\ \delta>0},
\end{align}
\textit{where the indices of }$s_i$ \textit{are taken modulo }3.
\end{lemma}
\begin{proof}
Let $(V,c)$ be a pair which satisfies (\ref{nphcondition2}) with respect to a non-collinear triple $(s_0,s_1,s_2)$ in $T$ and distinct elements $z_i\in V\setminus T\ (i=0,\ldots,5)$. We show NP-hardness of $\overrightarrow{\textbf{0}}$\textbf{-Ext}$[\mu]$ by constructing a gadget based on $(V,c)$ which satisfies (\ref{nph-condition}). Let $\mu_i:=\mu(s_{i-1},s_{i+1})+\mu(s_{i+1},s_{i-1})$ and $h_i:=(\mu_{i-1}+\mu_{i+1}-\mu_i)/2$ for $i=0,1,2$ (the indices of $\mu_i$ are taken modulo 3). Then we define a function $c':V\times V\rightarrow \mathbb{Q}_+$ as follows (see Figure 5):
\begin{align}
    &c'(z_i,z_{i+1})=c'(z_{i+1},z_i):=h_{i-1}\mathrm{\ \ \ }(0\leq i\leq 5).
\end{align}
\begin{figure}[tbp]
\begin{center}
\begin{overpic}[width=14cm]{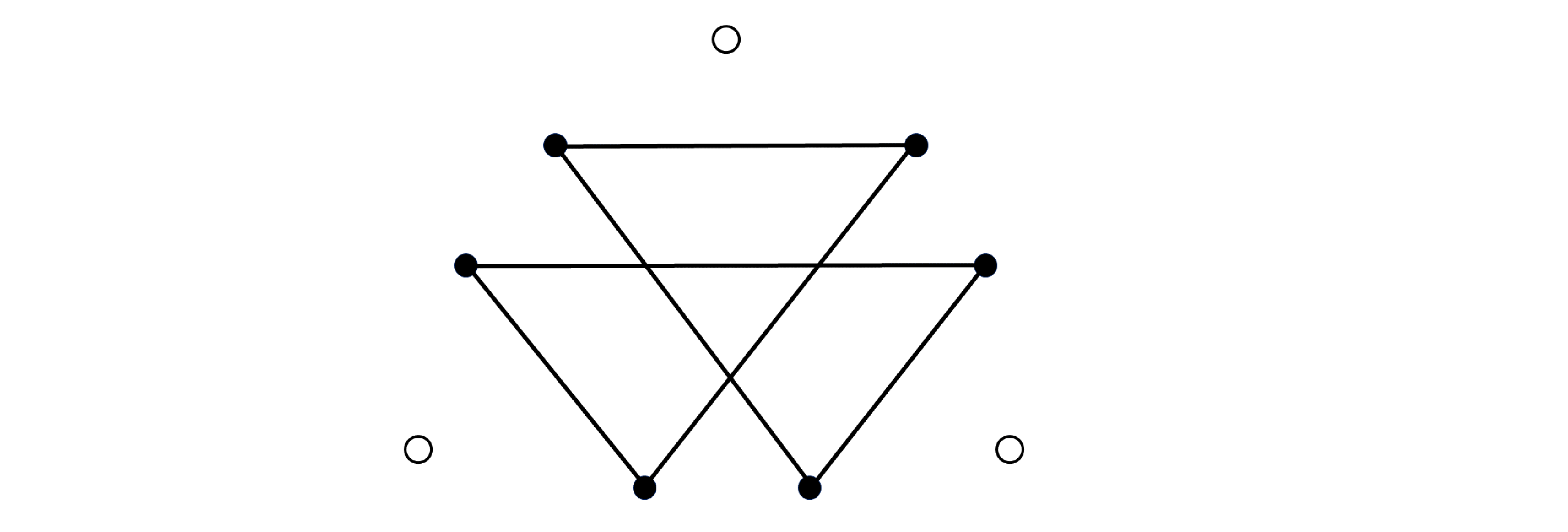}
\put(23,4){$s_0$}
\put(42.5,32){$s_1$}
\put(66,4){$s_2$}
\put(37,2){$z_1$}
\put(53,2){$z_4$}
\put(26,16.5){$z_2$}
\put(64,16.5){$z_3$}
\put(32,24){$z_5$}
\put(60,24){$z_0$}
\end{overpic}
\caption{the function $c'$}
\end{center}
\end{figure}Here the indices of $z_i$ are taken modulo 6, and the indices of $h_i$ are taken modulo 3. Also we define $c'(v_0,v_1):=0$ for all other pairs $(v_0,v_1)\in V\times V$ (undefined values of other functions below are also 0). Let $N$ be a sufficiently large positive rational (for example, $N:=1/\delta+4(h_0+h_1+h_2)\cdot\mathrm{max}\{\mu(s,t)\mid s,t\in T\}/\delta$). We define a function $\tilde{c}$ by $\tilde{c}:=Nc+c'$. Let $s:=s_0,t:=s_2,x:=z_1,y:=z_4$. We now show that the pair $(V,\tilde{c})$ satisfies (\ref{nph-condition}) with respect to $s,t,x,y$. For $r\neq s_{i-1},s_{i+1}$, the value $\tau(r,z_i)$ is so large that a map $\gamma$ with $\gamma(z_i)=r$ is not
optimal or nearly optimal, since $(V,c)$ satisfies (\ref{nphcondition2}) and $N$ is sufficiently large. We call such a map \textit{infeasible}. Therefore, it suffices to consider the case when $\gamma(z_i)\in \{s_{i-1},s_{i+1}\}$ for every $i$. Let $\rho:=2(h_0h_1+h_1h_2+h_2h_0),\ \alpha:=2\mathrm{min}\{h_0^2,h_1^2,h_2^2\}$. Without loss of generality, we may assume that $\alpha=2h_0^2$. In the case of $\gamma(x)=s$ and $\gamma(y)=t$, we have
\begin{align}
\label{eq:optimalassign1}
\tau_{\tilde{c}}(s_1,z_0|s_0,z_1|s_0,z_2|s_2,z_3|s_2,z_4|s_1,z_5)&=h_1\mu_1+h_2\mu_2+h_0\mu_0+N\tau_c^*\notag\\
&=h_1(h_2+h_0)+h_2(h_0+h_1)+h_0(h_1+h_2)+N\tau_c^*\notag\\
&=\rho+N\tau_c^*,
\end{align}
where $\tau_c^*$ is the optimal value of $\overrightarrow{\textbf{0}}$\textbf{-Ext}$[\mu]$ with respect to $(V,c)$. Similarly, in the case of $\gamma(x)=t$ and $\gamma(y)=s$, we have
\begin{align}
\label{eq:optimalassign2}
\tau_{\tilde{c}}(s_2,z_0|s_2,z_1|s_1,z_2|s_1,z_3|s_0,z_4|s_0,z_5)=\rho+N\tau_c^*.
\end{align}
On the other hand, in the case of $\gamma(x)=\gamma(y)=t$, we have
\begin{align}
    \tau_{\tilde{c}}(s_2,z_0|s_2,z_1|s_0,z_2|s_2,z_3|s_2,z_4|s_0,z_5)&=(h_0+h_0+h_1+h_1)\mu_1+N\tau_c^*\notag\\
    &=2(h_0+h_1)(h_2+h_0)+N\tau_c^*\notag\\
    &=\alpha+\rho+N\tau_c^*.
\end{align}
Also, in the case of $\gamma(x)=\gamma(y)=s$, we have
\begin{align}
    \tau_{\tilde{c}}(s_1,z_0|s_0,z_1|s_1,z_2|s_1,z_3|s_0,z_4|s_1,z_5)&=(h_2+h_2+h_0+h_0)\mu_2+N\tau_c^*\notag\\
    &=2(h_2+h_0)(h_0+h_1)+N\tau_c^*\notag\\
    &=\alpha+\rho+N\tau_c^*.
\end{align}

Let $\tau_{c}(\gamma)$ denote the value of the objective function of $\overrightarrow{\textbf{0}}$\textbf{-Ext}$[\mu]$ with a map $\gamma$, where the input is $(V,c)$. We simply denote $\tau(\gamma):=\tau_c(\gamma)$ if $c$ is clear in the context. To finish the proof, we show that $\tau_{\tilde{c}}(\gamma)\geq \alpha+\rho+N\tau_c^*$ if a map $\gamma$ is distinct from the assignments in (\ref{eq:optimalassign1}) and (\ref{eq:optimalassign2}). Let $\xi$ be the contribution to the value $\tau_{\tilde{c}}(\gamma)$ from $c'$. The value $\xi$ is represented by a sum of $h_ih_j\ (0\leq i,j\leq 2)$. Let $g_i$ be the contribution to the value $\tau_{\tilde{c}}(\gamma)$ from $c'(z_i,z_{i+1})$ and $c'(z_{i+1},z_i)$ for each $i$. We have the following four cases:
\begin{enumerate}[label=(\roman*),ref=\roman*]
    \item $\gamma(z_{i})=\gamma(z_{i+1})=s_{i-1}$,
    \label{nphcondition2:case1}
    \item $\gamma(z_{i})=s_{i+1},\ \gamma(z_{i+1})=s_{i}$,
    \label{nphcondition2:case2}
    \item $\gamma(z_{i})=s_{i+1},\ \gamma(z_{i+1})=s_{i-1}$,
    \label{nphcondition2:case3}
    \item $\gamma(z_{i})=s_{i-1},\ \gamma(z_{i+1})=s_{i}$.
    \label{nphcondition2:case4}
\end{enumerate}
We have $g_i=0,\ g_i=h_{i-1}h_{i}+h_{i-1}h_{i+1},\ g_i=h_{i-1}h_{i+1}+h_{i-1}^2,\ g_i=h_{i-1}h_i+h_{i-1}^2$ in the cases of (\ref{nphcondition2:case1}), (\ref{nphcondition2:case2}), (\ref{nphcondition2:case3}), (\ref{nphcondition2:case4}), respectively. We call a pair $z_iz_{i+1}$ \textit{slanting} if $z_i$ and $z_{i+1}$ satisfy (\ref{nphcondition2:case3}) or (\ref{nphcondition2:case4}). If $z_{i}z_{i+1}$ is not slanting for any $i$ with respect to $\gamma$, then $\gamma$ corresponds to the assignment in (\ref{eq:optimalassign1}) or (\ref{eq:optimalassign2}). If $z_iz_{i+1}$ is slanting for some $i\in\{0,\ldots,5\}$, then $z_jz_{j+1}$ is also slanting for another $j\in \{0,\ldots,5\}$. Hence, $\xi$ contains $h_{i-1}^2+h_{j-1}^2\geq \alpha$ in its representation. Also, focusing on the pairs $z_0z_1$ and $z_1z_2$, we can see that $g_0$ includes $h_2h_0$ in its representation when $\gamma(z_1)=s_0$, and $g_1$ includes $h_2h_0$ when $\gamma(z_1)=s_2$. Similarly, we can see that $g_i$ or $g_{i+1}$ includes $h_{i-1}h_i$ for each $i\in \{0,\ldots,5\}$. This completes the proof.
\end{proof}

We now show Theorem \ref{thm:directed-nph-extend2} for the case (\ref{thm:directed-nph-extend2-case1}) by use of Lemma \ref{lem:nphcondition2}. The proof we describe below is a directed version of that of Theorem \ref{thm:undirected-nph} (\ref{dichotomy:not modular}) in \cite{karzanov2004}. Let $\mu$ be a nonmodular rational-valued directed metric on $T$. For $x,y,z\in T$, we denote $\Delta(x,y,z):=\mu(x,y)+\mu(y,x)+\mu(y,z)+\mu(z,y)+\mu(z,x)+\mu(x,z)$. Let $(s_0,s_1,s_2)$ be a medianless triple such that $\Delta(s_0,s_1,s_2)$ is minimum. 
Let $\bar{\Delta}:=\Delta(s_0,s_1,s_2)$. Take six elements $z_0,z_1,\ldots,z_5$, and let $V:=T\cup \{z_0,z_1,\ldots,z_5\}$. Let $\mu_i:=\mu(s_{i-1},s_{i+1})+\mu(s_{i+1},s_{i-1})$ and $a_i:=(\mu_{i-1}+\mu_{i+1}-\mu_i)/\mu_{i-1}\mu_{i+1}$ for $i=0,1,2$, where the indices of $s_i$ and $\mu_i$ are taken modulo 3. Then we define a function $c:V\times V\rightarrow \mathbb{Q}_+$ as follows:
\begin{align}
    &c(s_i,z_{i+1})=c(z_{i+1},s_i)=1\mathrm{\ \ \ }(0\leq i\leq 5),\notag\\
    &c(s_i,z_{i+2})=c(z_{i+2},s_i)=1\mathrm{\ \ \ }(0\leq i\leq 5),
\end{align}
where the indices of $z_i$ are taken modulo 6. Also we define a function $c':V\times V\rightarrow \mathbb{Q}_+$ as follows:
\begin{align}    
&c'(s_i,z_j)=c'(z_j,s_i)=a_i\mathrm{\ \ \ }(0\leq i\leq 2,\ 0\leq j\leq 5).
\end{align}
Let $N$ be a sufficiently large positive rational. We define a function $\tilde{c}$ by $\tilde{c}:=Nc+c'$. We now show that the pair $(V,\tilde{c})$ satisfies (\ref{nphcondition2}). 
We first observe that $\tau_{\tilde{c}}(\gamma)$ is not the optimal or nearly optimal value if $\gamma(z_i)\notin I(s_{i-1},s_{i+1})\cap I(s_{i+1},s_{i-1})$ for some $i$. Consider the case when $\gamma(z_i)\in I(s_{i-1},s_{i+1})\cap I(s_{i+1},s_{i-1})$ holds for each $i$. We show the following claim:
\begin{claim}
\label{claim1}
\textit{Let} $x\in I(s_{i-1},s_{i+1})\cap I(s_{i+1},s_{i-1})$. \textit{Then at least one of the following conditions holds:}
\begin{enumerate}[label=(\roman*),ref=\roman*]
    \item \textit{Both of sequences} $(s_i,s_{i-1},x)$ \textit{and} $(x,s_{i-1},s_i)$ \textit{are} $\mu$\textit{-shortest.}
    \label{medianlesscondition1}
    \item \textit{Both of sequences} $(s_i,s_{i+1},x)$ \textit{and} $(x,s_{i+1},s_i)$ \textit{are} $\mu$\textit{-shortest.}
    \label{medianlesscondition2}
\end{enumerate}
\end{claim}
\begin{proof}
Suppose that (\ref{medianlesscondition2}) does not hold. Let $i=1$. By the assumption, we have $\mu(s_1,x)<\mu(s_1,s_2)+\mu(s_2,x)$ or $\mu(x,s_1)<\mu(x,s_2)+\mu(s_2,s_1)$. Then we have $\Delta(s_0,s_1,x)<\Delta(s_0,s_1,s_2)$. Hence, there exists a median $m$ of $s_0,s_1,x$. If $m=s_0$, then (\ref{medianlesscondition1}) holds. If $m\neq s_0$, then we have
\begin{align}
    \Delta(s_1,m,s_2)&=\mu(s_1,s_2)+\mu(s_2,s_1)+\mu(s_1,m)+\mu(m,s_1)+\mu(s_2,m)+\mu(m,s_2)\notag
    \\&<\mu(s_1,s_2)+\mu(s_2,s_1)+\mu(s_1,s_0)+\mu(s_0,s_1)+\mu(s_2,x)+\mu(x,m)+\mu(m,x)+\mu(x,s_2)\notag\\
    &< \Delta(s_0,s_1,s_2).
\end{align}
Hence, there exists a median $w$ of $s_1,m,s_2$. However, $w$ is also a median of $s_0,s_1,s_2$, and this is a contradiction.
\end{proof}

For each $i\in \{0,1,\ldots,5\}$, let $g_i$ be the contribution to the value $\tau_{\tilde{c}}(\gamma)$ from $c'(z_i,s_0)$, $c'(z_i,s_1)$, $c'(z_i,s_2)$, $c'(s_0,z_i)$, $c'(s_1,z_i)$, $c'(s_2,z_i)$. If $\gamma(z_i)=s_0$, then we have $g_i=a_1\mu_2+a_2\mu_1=(\mu_0+\mu_2-\mu_1)/\mu_0+(\mu_1+\mu_0-\mu_2)/\mu_0=2$. Similarly, we have $g_i=2$ when $\gamma(z_i)=s_1$ or $s_2$. We next consider the case when $\gamma(z_i)\in I(s_{i-1},s_{i+1})\cap I(s_{i+1},s_{i-1})\setminus \{s_{i-1},s_{i+1}\}$. Let $i=0$ and $\epsilon:=\mu(s_1,\gamma(z_0))+\mu(\gamma(z_0),s_1)$. By Claim~\ref{claim1}, we may assume that $\mu(s_0,\gamma(z_o))+\mu(\gamma(z_0),s_0)=\mu_2+\epsilon$ holds. Hence, we have
\begin{align}
    g_0&=a_0(\mu_2+\epsilon)+a_1\epsilon+a_2(\mu_0-\epsilon)\notag\\&=a_0\mu_2+a_2\mu_0+\epsilon(a_0+a_1-a_2)\notag\\
    &=2+\epsilon(a_0+a_1-a_2).
\end{align}
Note that we have
\begin{align}
    \mu_0\mu_1\mu_2(a_0+a_1-a_2)&=\mu_0(\mu_1+\mu_2-\mu_0)+\mu_1(\mu_0+\mu_2-\mu_1)-\mu_2(\mu_0+\mu_1-\mu_2)\notag\\
    &=2\mu_0\mu_1-\mu_1^2-\mu_0^2+\mu_2^2\notag\\
    &=\mu_2^2-(\mu_0-\mu_1)^2>0.
\end{align}
Hence, we have $g_0>2$. Similarly, for each $i\in \{0,1,\ldots,5\}$, we have $g_i>2$ if $\gamma(z_i)\in I(s_{i-1},s_{i+1})\cap I(s_{i+1},s_{i-1})\setminus \{s_{i-1},s_{i+1}\}$. Hence, the gadget $(V,\tilde{c})$ satisfies (\ref{nphcondition2}). 

\subsubsection{Proof of Theorem \ref{thm:directed-nph-extend2} for the case (\ref{thm:directed-nph-extend2-case2})}
Since $\mu$ is modular and not directed orbit-invariant, $H_\mu$ contains a directed orbit-varying modular cycle by Lemma \ref{lem:orbitvariant}. Let $(s_0,s_1,s_2,s_3,s_0)$ be a directed orbit-varying modular cycle in $H_\mu$. Take eight elements $x_0,x_1,y_0,y_1,z_0,z_1,w_0,w_1$, and let $V:=T\cup \{x_0,x_1,y_0,y_1,z_0,z_1,w_0,w_1\}$. Without loss of generality, we may assume that $\mu(s_0,s_1)-\mu(s_3,s_2)=\mu(s_2,s_3)-\mu(s_1,s_0)=\mu(s_0,s_3)-\mu(s_1,s_2)=\mu(s_2,s_1)-\mu(s_3,s_0)=k>0$. We define a function $c_4:V\times V\rightarrow \mathbb{Q}_+$ as follows (also see Figure 6):
\begin{align}
    &c_4(s_2,x_i)=c_4(x_i,s_2)=c_4(s_3,x_i)=c_4(x_i,s_3)=1\mathrm{\ \ \ }(i=0,1),\notag\\
    &c_4(s_0,z_i)=c_4(z_i,s_0)=c_4(s_1,z_i)=c_4(z_i,s_1)=1\mathrm{\ \ \ }(i=0,1),\notag\\
    &c_4(s_j,y_i)=c_4(y_i,s_j)=1\mathrm{\ \ \ }(i=0,1,\ j=0,1,2,3),\notag\\
    &c_4(s_1,w_0)=c_4(w_0,s_1)=c_4(s_2,w_0)=c_4(w_0,s_2)=1,\notag\\
    &c_4(s_0,w_1)=c_4(w_1,s_0)=c_4(s_3,w_1)=c_4(w_1,s_3)=1.
\end{align}
\begin{figure}[tbp]
\begin{center}
\begin{overpic}[width=14cm]{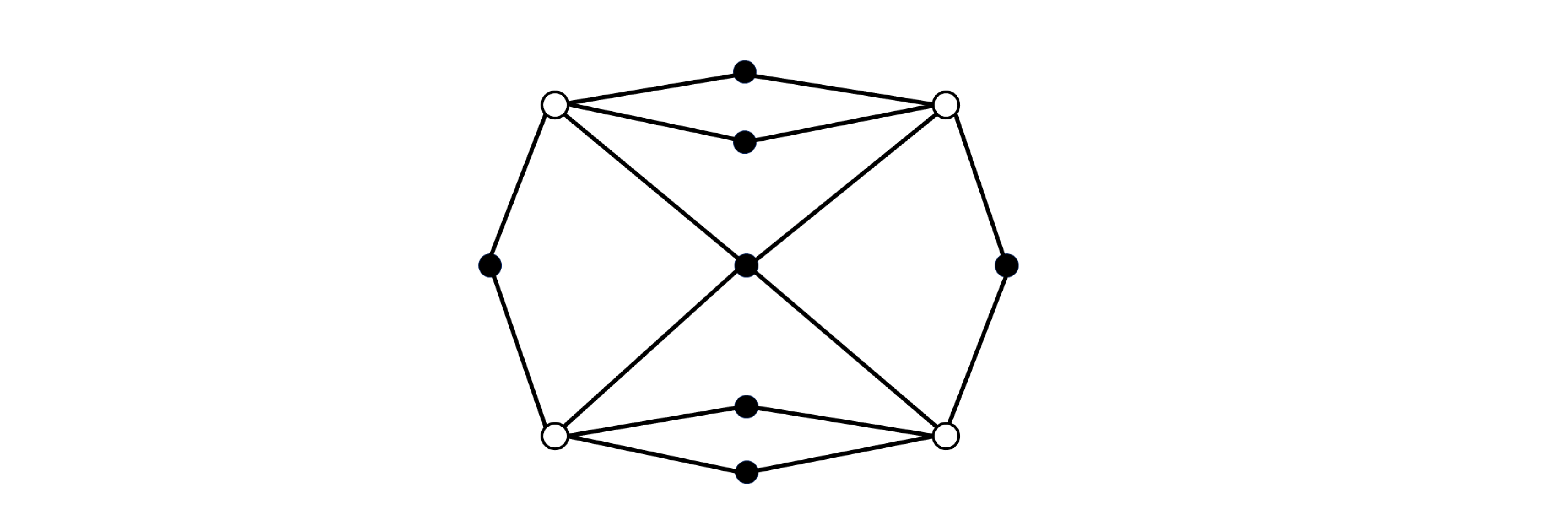}
\put(31.4,27){$s_0$}
\put(62,27){$s_1$}
\put(62,5){$s_2$}
\put(31.4,5){$s_3$}
\put(47,9.5){$x_0$}
\put(47,1){$x_1$}
\put(47,31.1){$z_0$}
\put(47,22){$z_1$}
\put(65.5,16.5){$w_0$}
\put(27,16.5){$w_1$}
\put(50,16.5){$y_i$}
\end{overpic}
\caption{the function $c_4$}
\end{center}
\end{figure}Let $\mu_{10}:=\mu(s_1,s_0),\ \mu_{12}:=\mu(s_1,s_2),\ \mu_{30}:=\mu(s_3,s_0),$ and $\mu_{32}:=\mu(s_3,s_2)$. We next define a function $c_3:V\times V\rightarrow \mathbb{Q}_+$ as follows:
\begin{align}
    &c_3(y_0,s_0)=k^2+\mu_{32}k+\mu_{32}\mu_{12},\notag\\
    &c_3(y_0,s_1)=(\mu_{10}+\mu_{12})k+\mu_{10}\mu_{12},\notag\\
    &c_3(y_0,s_2)=k^2+\mu_{10}k+\mu_{10}\mu_{30},\notag\\
    &c_3(y_0,s_3)=(\mu_{32}+\mu_{30})k+\mu_{32}\mu_{30},\notag\\
    &c_3(s_0,y_1)=(\mu_{10}+\mu_{30})k+\mu_{10}\mu_{30},\notag\\
    &c_3(s_1,y_1)=k^2+\mu_{32}k+\mu_{32}\mu_{30},\notag\\
    &c_3(s_2,y_1)=(\mu_{32}+\mu_{12})k+\mu_{32}\mu_{12},\notag\\
    &c_3(s_3,y_1)=k^2+\mu_{10}k+\mu_{10}\mu_{12}.
\end{align}
Then we define a function $c_2:V\times V\rightarrow \mathbb{Q}_+$ as follows:
\begin{align}
    &c_2(x_i,z_i)=c_2(z_i,x_i)=1\mathrm{\ \ \ }(i=0,1).
\end{align}
Also we define a function $c_1:V\times V\rightarrow \mathbb{Q}_+$ as follows:
\begin{align}
    &c_1(x_i,y_i)=c_1(y_i,x_i)=c_1(z_i,y_i)=c_1(y_i,z_i)=1\mathrm{\ \ \ }(i=0,1).
\end{align}
Finally, we define a function $c_0:V\times V\rightarrow \mathbb{Q}_+$ as follows:
\begin{align}
    &c_0(y_i,w_j)=c_0(w_j,y_i)=1\mathrm{\ \ \ }(i=0,1,\ j=0,1).
\end{align}
Let $N_1$ be a sufficiently large positive rational. In addition, let $N_i$ be a sufficiently large positive rational with respect to $N_{i-1}$ for $i=2,3,4$. We define a function $c$ by $c:=c_0+N_1c_1+N_2c_2+N_3c_3+N_4c_4$. We now show that the pair $(V,c)$ satisfies (\ref{nph-condition}) with respect to $s_2,s_3,x_0,x_1$. Focusing on the contribution from $N_4c_4$, we see that $\gamma$ is infeasible if $\gamma(x_i)\neq s_2,s_3$ for some $i$. Similarly, $\gamma$ is infeasible if $\gamma(z_i)\neq s_0,s_1$ holds for some $i$, or $\gamma(w_0)\neq s_1,s_2$ holds, or $\gamma(w_1)\neq s_0,s_3$ holds. Hence, it suffices to consider the case when $\gamma(x_i)\in \{s_2,s_3\}$ and $\gamma(z_i)\in \{s_0,s_1\}$ hold for each $i$, or the case when $\gamma(w_0)\in \{s_1,s_2\}$ and $\gamma(w_1)\in \{s_0,s_3\}$ hold. In addition, $\gamma$ is infeasible if $\gamma(y_i)\notin I(s_0,s_2)\cap I(s_2,s_0)\cap I(s_1,s_3)\cap I(s_3,s_1)$. Note that $H_\mu$ is modular by Lemma \ref{lemma:mu is modular -> H is modular}, which implies that edges $\{s_0,s_2\}$ and $\{s_1,s_3\}$ are not contained in $H_\mu$ due to the condition of Lemma~\ref{modularcondition}~(i). If $\gamma(y_i)\in I(s_0,s_2)\cap I(s_2,s_0)\cap I(s_1,s_3)\cap I(s_3,s_1)$ and $\gamma(y_i)\notin \{s_0,s_1,s_2,s_3\}$, then by Lemma \ref{metric->graph-shortest} (\ref{lemma14:case2}), $H_\mu$ has edges $\{\gamma(y_i),s_j\}$ for $j=0,1,2,3$. However, this contradicts the condition of Lemma \ref{modularcondition} (i). Therefore, it suffices to consider the case when $\gamma(y_i)\in \{s_0,s_1,s_2,s_3\}$. We next focus on the contribution from $N_3c_3$. Let $N_3\xi_0$ be the contribution to the value $\tau_{c}(\gamma)$ from $N_3c_3(y_0,s_0)$, $N_3c_3(y_0,s_1)$, $N_3c_3(y_0,s_2)$, $N_3c_3(y_0,s_3)$. If $\gamma(y_0)=s_0$, then we have
\begin{align}
    \xi_0&=(k^2+\mu_{10}k+\mu_{10}\mu_{30})(\mu_{32}+\mu_{12}+k)+((\mu_{10}+\mu_{12})k+\mu_{10}\mu_{12})(\mu_{32}+k)\notag\\&+((\mu_{32}+\mu_{30})k+\mu_{32}\mu_{30})(\mu_{12}+k)\notag\\
    &=k^3+(2\mu_{32}+2\mu_{10}+2\mu_{12}+\mu_{30})k^2\notag\\&+(2\mu_{32}\mu_{10}+\mu_{10}\mu_{30}+\mu_{32}\mu_{30}+2\mu_{10}\mu_{12}+2\mu_{32}\mu_{12}+\mu_{12}\mu_{30})k\notag\\
    &+\mu_{32}\mu_{10}\mu_{12}+\mu_{32}\mu_{12}\mu_{30}+\mu_{32}\mu_{10}\mu_{30}+\mu_{10}\mu_{12}\mu_{30}.
\end{align}
If $\gamma(y_0)=s_1$, then we have
\begin{align}
    \xi_0&=((\mu_{32}+\mu_{30})k+\mu_{32}\mu_{30})(\mu_{10}+\mu_{12}+k)+(k^2+\mu_{32}k+\mu_{32}\mu_{12})\mu_{10}+(k^2+\mu_{10}k+\mu_{10}\mu_{30})\mu_{12}\notag\\
    &=(\mu_{32}+\mu_{10}+\mu_{12}+\mu_{30})k^2+(2\mu_{32}\mu_{10}+\mu_{10}\mu_{30}+\mu_{32}\mu_{30}+\mu_{10}\mu_{12}+\mu_{32}\mu_{12}+\mu_{12}\mu_{30})k\notag\\
    &+\mu_{32}\mu_{10}\mu_{12}+\mu_{32}\mu_{12}\mu_{30}+\mu_{32}\mu_{10}\mu_{30}+\mu_{10}\mu_{12}\mu_{30}.
\end{align}
If $\gamma(y_0)=s_2$, then we have
\begin{align}
    \xi_0&=(k^2+\mu_{32}k+\mu_{32}\mu_{12})(\mu_{10}+\mu_{30}+k)+((\mu_{32}+\mu_{30})k+\mu_{32}\mu_{30})(\mu_{10}+k)\notag\\&+((\mu_{10}+\mu_{12})k+\mu_{10}\mu_{12})(\mu_{30}+k)\notag\\
    &=k^3+(2\mu_{32}+2\mu_{10}+\mu_{12}+2\mu_{30})k^2\notag\\
    &+(2\mu_{32}\mu_{10}+2\mu_{10}\mu_{30}+2\mu_{32}\mu_{30}+\mu_{10}\mu_{12}+\mu_{32}\mu_{12}+\mu_{12}\mu_{30})k\notag\\
    &+\mu_{32}\mu_{10}\mu_{12}+\mu_{32}\mu_{12}\mu_{30}+\mu_{32}\mu_{10}\mu_{30}+\mu_{10}\mu_{12}\mu_{30}.
\end{align}
If $\gamma(y_0)=s_3$, then we have
\begin{align}
    \xi_0&=((\mu_{10}+\mu_{12})k+\mu_{10}\mu_{12})(\mu_{32}+\mu_{30}+k)+(k^2+\mu_{32}k+\mu_{32}\mu_{12})\mu_{30}+(k^2+\mu_{10}k+\mu_{10}\mu_{30})\mu_{32}\notag\\
    &=(\mu_{32}+\mu_{10}+\mu_{12}+\mu_{30})k^2+(2\mu_{32}\mu_{10}+\mu_{10}\mu_{30}+\mu_{32}\mu_{30}+\mu_{10}\mu_{12}+\mu_{32}\mu_{12}+\mu_{12}\mu_{30})k\notag\\
    &+\mu_{32}\mu_{10}\mu_{12}+\mu_{32}\mu_{12}\mu_{30}+\mu_{32}\mu_{10}\mu_{30}+\mu_{10}\mu_{12}\mu_{30}.
\end{align}
Hence, we see that the value $\xi_0$ in the case of $\gamma(y_0)\in \{s_1,s_3\}$ is smaller than that in the case of $\gamma(y_0)\in \{s_0,s_2\}$. Therefore, it suffices to consider the case when $\gamma(y_0)\in \{s_1,s_3\}$. Similar to $N_3\xi_0$, let $N_3\xi_1$ be the contribution to the value $\tau_{c}(\gamma)$ from $N_3c_3(s_0,y_1)$, $N_3c_3(s_1,y_1)$, $N_3c_3(s_2,y_1)$, $N_3c_3(s_3,y_1)$. If $\gamma(y_1)=s_0$, then we have
\begin{align}
    \xi_1&=((\mu_{32}+\mu_{12})k+\mu_{32}\mu_{12})(\mu_{10}+\mu_{30}+k)+(k^2+\mu_{32}k+\mu_{32}\mu_{30})\mu_{10}+(k^2+\mu_{10}k+\mu_{10}\mu_{12})\mu_{30}\notag\\
    &=(\mu_{32}+\mu_{10}+\mu_{12}+\mu_{30})k^2+(2\mu_{32}\mu_{10}+\mu_{10}\mu_{30}+\mu_{32}\mu_{30}+\mu_{10}\mu_{12}+\mu_{32}\mu_{12}+\mu_{12}\mu_{30})k\notag\\&+\mu_{32}\mu_{10}\mu_{12}+\mu_{32}\mu_{12}\mu_{30}+\mu_{32}\mu_{10}\mu_{30}+\mu_{10}\mu_{12}\mu_{30}.
\end{align}
If $\gamma(y_1)=s_1$, then we have
\begin{align}
    \xi_1&=(k^2+\mu_{10}k+\mu_{10}\mu_{12})(\mu_{32}+\mu_{30}+k)+((\mu_{32}+\mu_{12})k+\mu_{32}\mu_{12})(\mu_{30}+k)\notag\\&+((\mu_{10}+\mu_{30})k+\mu_{10}\mu_{30})(\mu_{32}+k)\notag\\
    &=k^3+(2\mu_{32}+2\mu_{10}+\mu_{12}+2\mu_{30})k^2\notag\\
    &+(2\mu_{32}\mu_{10}+2\mu_{10}\mu_{30}+2\mu_{32}\mu_{30}+\mu_{10}\mu_{12}+\mu_{32}\mu_{12}+\mu_{12}\mu_{30})k\notag\\
    &+\mu_{32}\mu_{10}\mu_{12}+\mu_{32}\mu_{12}\mu_{30}+\mu_{32}\mu_{10}\mu_{30}+\mu_{10}\mu_{12}\mu_{30}.
\end{align}
If $\gamma(y_1)=s_2$, then we have
\begin{align}
    \xi_1&=((\mu_{10}+\mu_{30})k+\mu_{10}\mu_{30})(\mu_{32}+\mu_{12}k)+(k^2+\mu_{10}k+\mu_{10}\mu_{12})\mu_{32}+(k^2+\mu_{32}k+\mu_{32}\mu_{30})\mu_{12}\notag\\
    &=(\mu_{32}+\mu_{10}+\mu_{12}+\mu_{30})k^2+(2\mu_{32}\mu_{10}+\mu_{10}\mu_{30}+\mu_{32}\mu_{30}+\mu_{10}\mu_{12}+\mu_{32}\mu_{12}+\mu_{12}\mu_{30})k\notag\\
    &+\mu_{32}\mu_{10}\mu_{12}+\mu_{32}\mu_{12}\mu_{30}+\mu_{32}\mu_{10}\mu_{30}+\mu_{10}\mu_{12}\mu_{30}.
\end{align}
If $\gamma(y_1)=s_3$, then we have
\begin{align}
    \xi_1&=(k^2+\mu_{32}k+\mu_{32}\mu_{30})(\mu_{10}+\mu_{12}+k)+((\mu_{32}+\mu_{12})k+\mu_{32}\mu_{12})(\mu_{10}+k)\notag\\&+((\mu_{10}+\mu_{30})k+\mu_{10}\mu_{30})(\mu_{12}+k)\notag\\
    &=k^3+(2\mu_{32}+2\mu_{10}+2\mu_{12}+\mu_{30})k^2\notag\\
    &+(2\mu_{32}\mu_{10}+\mu_{10}\mu_{30}+\mu_{32}\mu_{30}+2\mu_{10}\mu_{12}+2\mu_{32}\mu_{12}+\mu_{12}\mu_{30})k\notag\\
    &+\mu_{32}\mu_{10}\mu_{12}+\mu_{32}\mu_{12}\mu_{30}+\mu_{32}\mu_{10}\mu_{30}+\mu_{10}\mu_{12}\mu_{30}.
\end{align}
Hence, we see that the value $\xi_1$ in the case of $\gamma(y_1)\in \{s_0,s_2\}$ is smaller than that in the case of $\gamma(y_1)\in \{s_1,s_3\}$. Therefore, it suffices to consider the case when $\gamma(y_1)\in \{s_0,s_2\}$. Recall that we consider the case when $\gamma(x_i)\in \{s_2,s_3\}$ and $\gamma(z_i)\in \{s_0,s_1\}$. Focus on the contribution from $N_2c_2$. Then we see that the value $\tau_{c}(\gamma)$ is not optimal or nearly optimal if $\gamma(x_i)=s_2$ and $\gamma(z_i)=s_0$, or if $\gamma(x_i)=s_3$ and $\gamma(z_i)=s_1$. Hence, it suffices to consider the case when $\gamma(x_i)=s_2$ and $\gamma(z_i)=s_1$ hold, or the case when $\gamma(x_i)=s_3$ and $\gamma(z_i)=s_0$ hold. We next focus on the contribution from $N_1c_1$. Then we see that a map $\gamma$ is infeasible if $\gamma(x_0)=s_2,\ \gamma(z_0)=s_1,$ and $\gamma(y_0)=s_3$ hold. Similarly, we see that a map $\gamma$ is infeasible if $\gamma(x_0)=s_3,\ \gamma(z_0)=s_0,$ and $\gamma(y_0)=s_1$ hold. Hence, we only consider the case when $\gamma(x_0)=s_2,\ \gamma(z_0)=s_1,$ and $\gamma(y_0)=s_1$ hold, or the case when $\gamma(x_0)=s_3,\ \gamma(z_0)=s_0,$ and $\gamma(y_0)=s_3$ hold. Applying the similar argument to $x_1$, $z_1$, and $y_1$, we see that it suffices to consider the case when $\gamma(x_1)=s_2,\ \gamma(z_1)=s_1,$ and $\gamma(y_1)=s_2$ hold, or the case when $\gamma(x_1)=s_3,\ \gamma(z_1)=s_0,$ and $\gamma(y_1)=s_0$ hold. We finally focus on the contribution from $c_0$. Let $\sigma$ be the contribution to $\tau_{c}(\gamma)$ from $c_0$. Consider the case when $\gamma(x_0)=\gamma(x_1)=s_2$, $\gamma(y_0)=s_1$, $\gamma(y_1)=s_2$. Then we have $\sigma=2(\mu_{32}+\mu_{10}+\mu_{12}+\mu_{30}+2k)$ for any $\gamma(w_0)\in \{s_1,s_2\},\gamma(w_1)\in \{s_0,s_3\}$. Similarly, if $\gamma(x_0)=\gamma(x_1)=s_3$, $\gamma(y_0)=s_3$, $\gamma(y_1)=s_0$, then we have $\sigma=2(\mu_{32}+\mu_{10}+\mu_{12}+\mu_{30}+2k)$ for any $\gamma(w_0)\in \{s_1,s_2\},\gamma(w_1)\in \{s_0,s_3\}$. On the other hand, if $\gamma(x_0)=s_2$, $\gamma(x_1)=s_3$, $\gamma(y_0)=s_1$, $\gamma(y_1)=s_0$, then $\sigma$ takes the minimum value $\sigma=2(\mu_{32}+\mu_{10}+k)$ when $\gamma(w_0)=s_1$ and $\gamma(w_1)=s_0$. Similarly, if $\gamma(x_0)=s_3$, $\gamma(x_1)=s_2$, $\gamma(y_0)=s_3$, $\gamma(y_1)=s_2$, then $\sigma$ takes the minimum value $\sigma=2(\mu_{32}+\mu_{10}+k)$ when $\gamma(w_0)=s_2$ and $\gamma(w_1)=s_3$. Thus, the pair $(V,c)$ satisfies the condition (\ref{nph-condition}).

\subsubsection{Proof of Theorem \ref{thm:directed-nph-extend2} for the case (\ref{thm:directed-nph-extend2-case3})}
We extend the proof of Theorem \ref{thm:undirected-nph} for the case (\ref{dichotomy: not orientable}) in \cite{karzanov2004} to that of Theorem \ref{thm:directed-nph-extend2} for the case (\ref{thm:directed-nph-extend2-case3}). Since the underlying graph $H_\mu$ is not orientable, there exists a sequence $(\overrightarrow{e_0},\overrightarrow{e_1},\ldots,\overrightarrow{e_k})$ ($\overrightarrow{e_i}=(s_i,t_i)$ is an oriented edge of $H_\mu$ for each $i$) such that  $H_\mu$ contains a 4-cycle $(s_i,t_i,t_{i+1},s_{i+1},s_i)$ for each $i\in \{0,\ldots,k-1\}$, and $t_k=s_0,s_k=t_0$. Then we have $h:=\mu(s_0,t_0)=\mu(s_1,t_1)=\cdots=\mu(s_k,t_k)=\mu(t_0,s_0)=\mu(t_1,s_1)=\cdots=\mu(t_k,s_k)$, and $f_i:=\mu(s_i,s_{i+1})=\mu(t_i,t_{i+1})$, $g_i:=\mu(s_{i+1},s_i)=\mu(t_{i+1},t_i)$ for $i=0,\ldots,k-1$, since $\mu$ is directed orbit-invariant. Take $2k$ elements $z_0,z_1,\ldots,z_{2k-1}$, and let $V:=T\cup \{z_0,z_1,\ldots,z_{2k-1}\}$. We now define a function $c:V\times V\rightarrow \mathbb{Q}_+$ as follows:
\begin{align}
    &c(z_i,s_i)=c(s_i,z_i)=c(z_i,t_i)=c(t_i,z_i)=1\mathrm{\ \ \ }(i=0,1,\ldots,2k-1),
\end{align}
where the indices of $s_i$ and $t_i$ are taken modulo $k$. We also define a function $c':V\times V\rightarrow \mathbb{Q}_+$ as follows (see Figure 7):
\begin{align}
    &c'(z_i,z_{i+1})=c'(z_{i+1},z_i)=1\mathrm{\ \ \ }(i=0,1,\ldots,2k-1),
\end{align}
where the indices of $z_i$ are taken modulo $2k$.
\begin{figure}[tbp]
\begin{center}
\begin{overpic}[width=14cm]{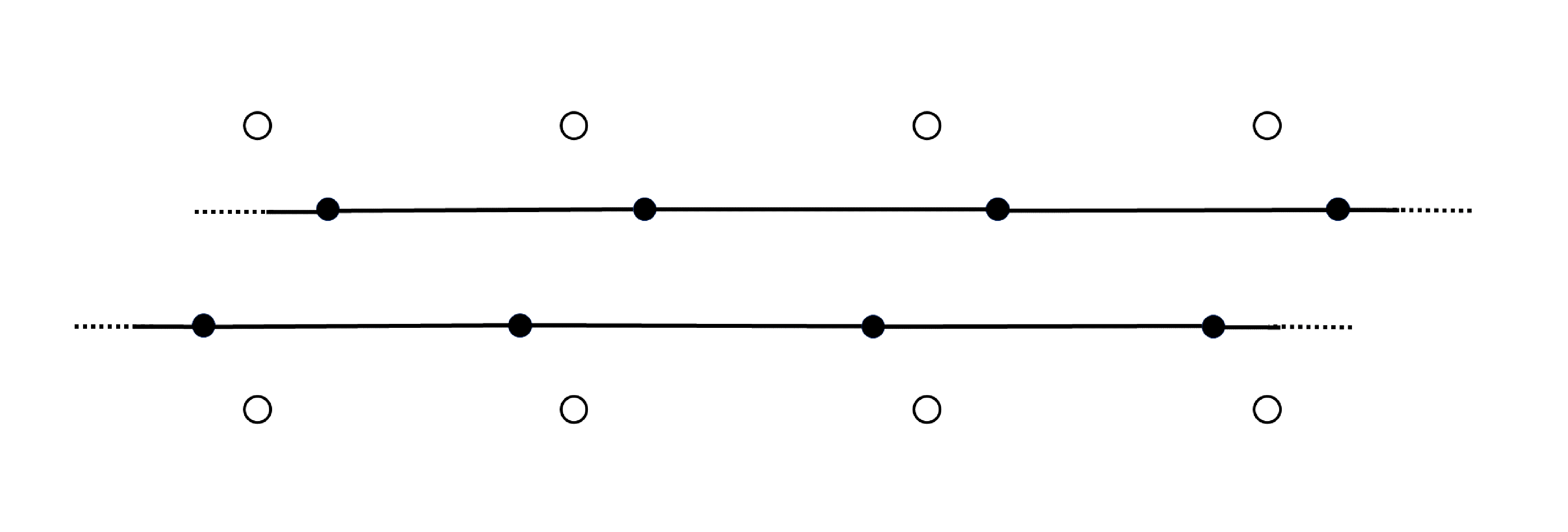}
\put(15.5,28.5){$s_0$}
\put(36,28.5){$s_1$}
\put(58,28.5){$s_2$}
\put(80,28.5){$s_3$}
\put(15.5,4.2){$t_0$}
\put(36,4.2){$t_1$}
\put(58,4.2){$t_2$}
\put(80,4.2){$t_3$}
\put(21.7,22){$z_0$}
\put(42,22){$z_1$}
\put(64.3,22){$z_2$}
\put(86,22){$z_3$}
\put(10,10.5){$z_4$}
\put(30,10.5){$z_5$}
\put(53,10.5){$z_6$}
\put(75,10.5){$z_7$}
\put(95,19.6){to $z_4$}
\put(86.8,12.5){to $z_0$}
\end{overpic}
\caption{the function $c'$ (the case of $k=4$)}
\end{center}
\end{figure}
Let $N$ be a sufficiently large positive rational. We define a function $\tilde{c}$ by $\tilde{c}:=Nc+c'$. Then we show that the pair $(V,\tilde{c})$ satisfies (\ref{nph-condition}) with respect to $s_0,t_0,z_0,z_k$. Let $s_{i+k}:=t_i\ (i=0,\ldots,k-1)$, and the indices of $s_i$ are taken modulo $2k$ below. We first observe that $\gamma$ is infeasible if $\gamma(z_i)\neq s_i,s_{i+k}$ for some $i$, due to the contribution from $Nc$. Thus, it suffices to consider the case when $\gamma(z_i)\in \{s_i,s_{i+k}\}$ for every $i\in \{0,\ldots,2k-1\}$. Let $\sigma$ be the contribution to $\tau_{\tilde{c}}(\gamma)$ from $c'$, and $\sigma_{i}$ be the contribution from $c'(z_i,z_{i+1})$ and $c'(z_{i+1},z_i)$ for each~$i$. If $(\gamma(z_i),\gamma(z_{i+1}))=(s_i,s_{i+1})$ or $(s_{i+k},s_{i+k+1})$, then we have $\sigma_i=f_i+g_i$. Otherwise, we have $\sigma_i=f_i+g_i+2h$. Consider the case when $\gamma(z_0)=s_0$ and $\gamma(z_k)=s_k$. In this case, if $\gamma(z_i)=s_i$ holds for $i=0,\ldots,2k-1$, then we have $\sigma_i=f_i+g_i$ for each $i$, and $\sigma=2\sum_{i=0}^{k-1}(f_i+g_i)$. Similarly, in the case of $\gamma(z_0)=s_k$ and $\gamma(z_k)=s_0$, we have $\sigma=2\sum_{i=0}^{k-1}(f_i+g_i)$ if $\gamma(z_i)=s_{i+k}$ for $i=0,\ldots,2k-1$. Consider the case when $\gamma(z_0)=\gamma(z_k)=s_0$. In this case, there exist two or more integers $i\in \{0,\ldots,2k-1\}$ for any $\gamma$ such that $\sigma_i=f_i+g_i+2h$ holds, and exactly two integers for some $\gamma$. Hence, the minimum value of $\sigma$ is $2\sum_{i=0}^{k-1}(f_i+g_i)+4h$. In the case of $\gamma(z_0)=\gamma(z_k)=s_k$, we also see that the minimum value of $\sigma$ is $2\sum_{i=0}^{k-1}(f_i+g_i)+4h$ by the similar argument.

\subsubsection{Proof of Theorem \ref{thm:directed-nph-new}}
Let $(s_0,s_1,s_2)$ be a biased non-collinear triple in $T$. For $i=0,1,2$, since $(s_i,s_{i+1})$ is a biased pair, $R_{\mu}(s_i,x)>R_{\mu}(x,s_{i+1})$ holds for every $x\in I(s_i,s_{i+1})\cap I(s_{i+1},s_i)\setminus \{s_i,s_{i+1}\}$, or $R_{\mu}(s_i,x)<R_{\mu}(x,s_{i+1})$ holds for every $x\in I(s_i,s_{i+1})\cap I(s_{i+1},s_i)\setminus \{s_i,s_{i+1}\}$, where the indices of $s_i$ are taken modulo~3. 
Take six elements $z_0,z_1,\ldots,z_5$, and let $V:=T\cup \{z_0,\ldots,z_5\}$. We define a function $c:V\times V\rightarrow \mathbb{Q}_+$ as follows:
\begin{align}
    &c(s_{i-1},z_i)=c(z_i,s_{i-1})=c(s_{i+1},z_i)=c(z_i,s_{i+1})=1\mathrm{\ \ \ }(i=0,\ldots,5).
\end{align}
We next define a function $c':V\times V\rightarrow \mathbb{Q}_+$ as follows. For each $i\in \{0,\ldots,5\}$,
\begin{itemize}
  \item if $R_{\mu}(s_{i-1},x)>R_{\mu}(x,s_{i+1})$ holds for every $x\in I(s_{i-1},s_{i+1})\cap I(s_{i+1},s_{i-1})\setminus \{s_{i-1},s_{i+1}\}$, then 
\begin{align}
    &c'(s_{i-1},z_i)=\mu(s_{i+1},s_{i-1}),\notag\\
    &c'(s_{i+1},z_i)=\mu(s_{i-1},s_{i+1}).
\end{align}
 \item if $R_{\mu}(s_{i-1},x)<R_{\mu}(x,s_{i+1})$ holds for every $x\in I(s_{i-1},s_{i+1})\cap I(s_{i+1},s_{i-1})\setminus \{s_{i-1},s_{i+1}\}$, then 
\begin{align}
    &c'(z_i,s_{i-1})=\mu(s_{i-1},s_{i+1}),\notag\\
    &c'(z_i,s_{i+1})=\mu(s_{i+1},s_{i-1}).
\end{align}
\end{itemize}
Let $N$ be a sufficiently large positive rational. We define a function $\tilde{c}$ by $\tilde{c}:=Nc+c'$. We now show that the pair $(V,\tilde{c})$ satisfies the condition (\ref{nphcondition2}) with respect to $s_0,s_1,s_2,z_0,z_1,z_2,z_3,z_4,z_5$. Focusing on the contribution from $Nc$, we see that a map $\gamma$ is infeasible if $\gamma(z_i)\notin I(s_{i-1},s_{i+1})\cap I(s_{i+1},s_{i-1})$ holds for some $i$. Thus, it suffices to consider the case when $\gamma(z_i)\in I(s_{i-1},s_{i+1})\cap I(s_{i+1},s_{i-1})$ holds for each $i$. We next focus on the contribution from $c'$. Consider the case when $R_{\mu}(s_{i-1},x)>R_{\mu}(x,s_{i+1})$ holds for every $x\in I(s_{i-1},s_{i+1})\cap I(s_{i+1},s_{i-1})\setminus \{s_{i-1},s_{i+1}\}$. Let $\sigma_i$ be the contribution to $\tau_{\tilde{c}}(\gamma)$ from $c'(s_{i-1},z_i)$ and $c'(s_{i+1},z_i)$. If $\gamma(z_i)\in \{s_{i-1},s_{i+1}\}$, then we have $\sigma_i=\mu(s_{i-1},s_{i+1})\mu(s_{i+1},s_{i-1})$. If $\gamma(z_i)\in I(s_{i-1},s_{i+1})\cap I(s_{i+1},s_{i-1})\setminus \{s_{i-1},s_{i+1}\}$, then we have
\begin{align}
    \sigma_i&=\mu(s_{i+1},s_{i-1})\mu(s_{i-1},\gamma(z_i))+\mu(s_{i-1},s_{i+1})\mu(s_{i+1},\gamma(z_i))\notag\\
    &>\mu(s_{i+1},s_{i-1})\mu(s_{i-1},\gamma(z_i))+\mu(\gamma(z_i),s_{i+1})\mu(s_{i+1},s_{i-1})\notag\\
    &\geq \mu(s_{i-1},s_{i+1})\mu(s_{i+1},s_{i-1}). 
\end{align}

Consider the case when $R_{\mu}(s_{i-1},x)<R_{\mu}(x,s_{i+1})$ holds for every $x\in I(s_{i-1},s_{i+1})\cap I(s_{i+1},s_{i-1})\setminus \{s_{i-1},s_{i+1}\}$. Let $\sigma_i'$ be the contribution to $\tau_{\tilde{c}}(\gamma)$ from $c'(z_i,s_{i-1})$ and $c'(z_i,s_{i+1})$. Similar to the above case, we have $\sigma_i'=\mu(s_{i-1},s_{i+1})\mu(s_{i+1},s_{i-1})$ if $\gamma(z_i)\in \{s_{i-1},s_{i+1}\}$, and we have $\sigma_i'>\mu(s_{i-1},s_{i+1})\mu(s_{i+1},s_{i-1})$ if $\gamma(z_i)\in I(s_{i-1},s_{i+1})\cap I(s_{i+1},s_{i-1})\setminus \{s_{i-1},s_{i+1}\}$. Thus, $(V,\tilde{c})$ satisfies the condition (\ref{nphcondition2}). This completes the proof.

\section*{Acknowledgments}
We thank the referees for helpful comments. The first author was supported by JSPS KAKENHI Grant Numbers JP17K00029
and JST PRESTO Grant Number JPMJPR192A, Japan.

\bibliography{main}
\bibliographystyle{plain}
\end{document}